\documentclass[compsoc,conference,a4paper,10pt,times]{IEEEtran}
\IEEEoverridecommandlockouts
\IEEEpubid{\begin{minipage}{\textwidth}
	\center
\copyright 2021 IEEE.
Personal use of this material is permitted.
Permission from IEEE must be obtained for all other uses, in any current or future media, including reprinting/republishing this material for advertising or promotional purposes, creating new collective works, for resale or redistribution to servers or lists, or reuse of any copyrighted component of this work in other works.
\end{minipage}}
\usepackage{booktabs}
\usepackage{cite}
\usepackage{amsmath}
\usepackage{amssymb}
\usepackage{amsfonts}
\usepackage{amsthm}
\usepackage{graphicx}
\usepackage{textcomp}
\usepackage{bmpsize}
\usepackage{xcolor}
\usepackage{lipsum}
\usepackage[colorlinks=true,urlcolor=black]{hyperref}
\usepackage{graphicx}
\usepackage[noend]{algpseudocode}
\usepackage{ifthen}
\usepackage[inline]{enumitem}
\usepackage{xspace}
\usepackage{xfrac}
\def\BibTeX{{\rm B\kern-.05em{\sc i\kern-.025em b}\kern-.08em
    T\kern-.1667em\lower.7ex\hbox{E}\kern-.125emX}}
\usepackage{pifont}%

\usepackage{multirow}
\usepackage{epstopdf}
\usepackage{subcaption}
\usepackage{tikz}
\usetikzlibrary{positioning,decorations.pathreplacing,fit,calc,arrows.meta,intersections,}
\pgfdeclarelayer{background}
\pgfsetlayers{background,main}

\newcommand{\toolname}{\textsc{BGPeek-a-Boo}\xspace}
\newtheorem{theorem}{Theorem}
\newtheorem{definition}{Definition}

\newcommand{\amppot}{\textsc{AmpPot}}
\newcommand{\constrained}{\emph{constrained}\xspace}
\newcommand{\unconstrained}{\emph{unconstrained}\xspace}
\newcommand{\rCPv}{\emph{customer-provider}\xspace}
\newcommand{\rPCv}{\emph{provider-customer}\xspace}
\newcommand{\rPPv}{\emph{peer-to-peer}\xspace}
\newcommand{\rCP}{\emph{CP}\xspace}

\newcommand{\rPP}{\emph{P2P}\xspace}

\DeclareMathOperator{\asn}{asn}

\DeclareMathOperator{\dominatees}{dominatees}

\DeclareMathOperator{\reachable}{reachable}
\DeclareMathOperator{\indChange}{indChange}
\DeclareMathOperator{\indStop}{indStop}
\newcommand{\AS}{\ensuremath\mathrm{AS}}
\newcommand{\CP}{\ensuremath\mathrm{CP}}
\newcommand{\PP}{\ensuremath\mathrm{P2P}}
\newcommand{\nodepath}{\ensuremath\pi}
\newcommand{\aspath}{\ensuremath\Pi}

\newcommand{\rp}{r_{\textrm{passive}}}
\newcommand{\ra}{\vec{r}_{\textrm{active}}}
\newcommand{\noeffect}{\texttt{NO\_EFFECT}}
\newcommand{\ttlchange}{\texttt{TTL\_CHANGE}}
\newcommand{\trafficstop}{\texttt{STOP}}

\newcommand{\allAS}{\mathcal{A}}

\newcommand{\candidates}{\mathcal{C}}
\newcommand{\logbook}{\mathcal{L}}
\newcommand{\probe}{\mathcal{P}}
\newcommand{\exclude}{\probe_{\textrm{inconsistent}}}
\newcommand{\last}{L}

\newboolean{todos}
\setboolean{todos}{false}

\newboolean{blind}
\setboolean{blind}{false}

\ifthenelse{\boolean{todos}}{
	\usepackage{xcolor}
	\usepackage{pdfcomment}

}{}

\begin{document}

\title{BGPeek-a-Boo: Active BGP-based Traceback for Amplification DDoS Attacks}

\ifthenelse{\boolean{blind}}{
	\author{}
}{
	\author{\IEEEauthorblockN{Johannes Krupp}
	\IEEEauthorblockA{\textit{CISPA Helmholtz Center for Information Security} \\
	Saarbrücken, Saarland, Germany \\
	johannes.krupp@cispa.saarland}
	\and
	\IEEEauthorblockN{Christian Rossow}
	\IEEEauthorblockA{\textit{CISPA Helmholtz Center for Information Security} \\
	Saarbrücken, Saarland, Germany \\
	rossow@cispa.saarland}
	}
}

\maketitle

\begin{abstract}
	Amplification DDoS attacks inherently rely on IP spoofing to steer attack traffic to the victim.
	At the same time, IP spoofing undermines prosecution, as the originating attack infrastructure remains hidden. Researchers have therefore proposed various mechanisms to trace back amplification attacks (or IP-spoofed attacks in general).
	However, existing traceback techniques require either the cooperation of external parties or \emph{a priori} knowledge about the attacker.

	We propose \toolname, a BGP-based approach to trace back amplification attacks to their origin network.
	\toolname monitors amplification attacks with honeypots and uses \emph{BGP Poisoning} to temporarily shut down ingress traffic from selected Autonomous Systems.
	By systematically probing the entire AS space, we detect systems forwarding and originating spoofed traffic.
	We then show how a graph-based model of BGP route propagation can reduce the search space, resulting in a $5\times$ median speed-up and over $20\times$ for $\sfrac{1}{4}$ of all cases.
	\toolname achieves a unique traceback result $60\%$ of the time in a simulation-based evaluation supported by real-world experiments.
\end{abstract}


\section{Introduction}
Amplification attacks~\cite{rossow2014amplification} continue to be one the most powerful type of DDoS attacks, reaching attack bandwidths as high as $1.7$ Tbps in 2018~\cite{netscout17tbps} or $2.3$ Tbps in 2020~\cite{zdnet23tbps}.
These attacks rely on \emph{IP spoofing}:
As IP header information is not authenticated, crucial fields such as the packet's source address can be set to arbitrary values by the attacker.
Even worse, IP spoofing not only enables these attacks in the first place, it also effectively hides the attack's origin.
Without knowing the true network origin, identifying the actors behind these attacks is nigh impossible.
Thus, a traceback mechanism for these attacks is of prime importance.

Previous traceback approaches for amplification attacks can only link attacks to scanners used in attack preparation~\cite{krupp2016identifying} or re-identify attacks from known sources~\cite{krupp2017linking}.
This restricts them to the subset of incidents where rich auxiliary information is known \emph{a priori}.
Approaches to trace back IP spoofing in general~\cite{john2009ddos} are mostly based around the idea of packet marking~\cite{song2001advanced,duwairi2004efficient,dong2005efficient,shokri2006ddpm,belenky2007deterministic,gao2007practical}, where routers encode path information in the packet header, or collecting flow telemetry data~\cite{snoeren2001hash,snoeren2002single,li2004large,korkmaz2007single,sung2008large}.
However, both require the cooperation of a large number of routers along the path and thus a widespread deployment on the Internet---something we have not seen despite these approaches being known for over a decade.

In this paper we propose \toolname, a novel approach to trace back amplification attacks that requires neither cooperation of on-path routers nor knowledge of potential attack sources.
The main insight behind our approach is as follows:
While attackers may spoof IP level information, they are usually tightly coupled to a given spoofing-capable network location.
Packets sent by the attacker are thus bound to the routes chosen by their network provider.
This allows us to use the Border Gateway Protocol (BGP) to identify the Autonomous System (AS) that emits the spoofed attack traffic, which constitutes a fundamental step towards fighting these attacks:
Once identified, prosecutors can contact the AS operators to investigate the perpetrators behind the attack, which must be customers of the AS.
Further, the spoofing AS can be pressured into implementing egress filtering by its peers, similar to what happened to McColo in 2008 (cf.~\cite{mccolo}).

\toolname, shown in \autoref{fig:system}, consists of a number of amplification honeypots organized in multiple \texttt{/24} prefixes and a BGP router that can advertise routes for these prefixes.
The honeypots emulate services that are vulnerable to amplification in order to be selected as reflectors in amplification attacks~\cite{kramer2015amppot}.
During attacks, these honeypots will receive spoofed traffic sent by the attacker.
Through \emph{BGP Poisoning} we can then exclude certain ASes from propagating routes towards our system.
In particular, depriving the attacker of a route causes the spoofed traffic to either switch to an alternative route, which may be observed by a change in TTL values at the honeypots, or to cease entirely.
Building on this observation, we systematically probe ASes to uncover those involved in forwarding attack traffic---eventually leading us to the spoofing AS itself.

In a second step, we show how AS relationship data can be used to drastically limit the search space.
For this we build a \emph{BGP flow graph} that captures how BGP advertisements propagate and analyze which systems are \emph{reachable} and \emph{dominated} by others.
We find that both algorithms achieve a perfect attribution result $100\%$ of the time in an idealized, and still over $60\%$ in a more realistic simulation.
Our naive algorithm requires a median of $549$~BGP Poisoning steps ($91.5$~hours) for traceback, while our graph-based algorithm improves this to $98.5$~steps ($16.4$~hours), with $25\%$ of cases even terminating in at most $29$~steps ($4.8$~hours).
An 8-fold parallelization of our methodology reduces the median traceback duration to less than an hour.

In summary, our contributions are the following:
\begin{enumerate}
	\item We propose a novel approach for AS-level traceback of IP spoofing by leveraging BGP Poisoning.
		Our approach requires neither cooperation of external parties nor a priori knowledge about the attacker.
	\item We present two traceback algorithms, showing that BGP-based traceback is feasible in principle and can be greatly sped up when augmented with AS relationship data.
	\item We provide an extensive simulation-based evaluation, measuring the influence of various parameters on performance and correctness.
		We confirm our simulator through real-world experiments using the PEERING BGP testbed~\cite{schlinker2019peering} and RIPE Atlas~\cite{ripeatlas}.
\end{enumerate}

\begin{figure*}[t]
 \centering
 \begin{tikzpicture}[
   as/.style={circle,draw,minimum size=7mm,align=center}]

  \begin{pgfonlayer}{main}
   \def\xhoney1{6}
   \def\xlabel{8}

   \node [as] (AS0) at (0,-1.5) {$\bf V$};
   \node [as, text=gray] (AS2) at (1.25,-0.75) {$H$};
   \node [as, text=gray] (AS1) at (1.25,0.75) {$G$};

   \node [as, text=gray] (AS3) at (0,0) {$F$};
   \node [as, text=gray] (AS4) at (0,1.5) {$E$};

   \node [as, text=gray] (AS6) at (-1.25,-0.75) {$D$};

   \node [as, text=gray] (AS7) at (-1.25,0.75) {$C$};
   \node [as, text=gray] (AS8) at (-1.25,2.25) {$B$};

   \node [as] (AS9) at (-2.5,1.5) {$\bf A$};

   \node [label=below:BGP Router] (bgp) at (3.75,0) {\includegraphics{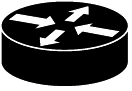}};

   \node [] (honeypot-probe-1-2) at ($(\xhoney1, 2)+(0.2,0.2)$) {\includegraphics[scale=0.75]{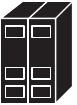}};
   \node [] (honeypot-probe-1-1) at (\xhoney1, 2) {\includegraphics[scale=0.75]{fig/standard_host}};

   \node [] (honeypot-probe-n-2) at ($(\xhoney1, 0.5)+(0.2,0.2)$) {\includegraphics[scale=0.75]{fig/standard_host}};
   \node [] (honeypot-probe-n-1) at (\xhoney1, 0.5) {\includegraphics[scale=0.75]{fig/standard_host}};

   \node [] (honeypot-control-2) at ($(\xhoney1, -1)+(0.2,0.2)$) {\includegraphics[scale=0.75]{fig/standard_host}};
   \node [label={[name=label,align=center] below:{Amplification\\Honeypots}}] (honeypot-control-1) at (\xhoney1, -1) {\includegraphics[scale=0.75]{fig/standard_host}};

  \node [label=center:{$\vdots$}] (probe-dots) at (\xhoney1,1.45) {};
  \node [label=center:{$\vdots$}] (net-probe-dots) at (\xlabel,1.35) {};

   \node [align=center] (net-probe-1) at (\xlabel,2) {Probe Prefix 1\\\texttt{X.Y.1.0/24}};
  \node [align=center] (net-probe-n) at (\xlabel,0.5) {Probe Prefix $n$\\\texttt{X.Y.n.0/24}};
  \node [align=center] (net-control) at (\xlabel,-1) {Control Prefix\\\texttt{X.Y.0.0/24}};

  \draw (AS1) -- (bgp);
  \draw (AS2) -- (bgp);
  \draw (AS3) -- (AS1);
  \draw (AS3) -- (AS2);
   \draw [name path=e-g] (AS4) -- (AS1);
  \draw (AS6) -- (AS3);
  \draw (AS7) -- (AS4);
  \draw (AS8) -- (AS4);
  \draw (AS9) -- (AS7);
  \draw (AS0) -- (AS2);

  \draw (bgp) -- (honeypot-probe-1-1);
  \draw (bgp) -- (honeypot-probe-n-1);
  \draw (bgp) -- (honeypot-control-1);

   \path[name path=cut, save path=\cut] (AS4) ++(280:7mm) arc [start angle=280, end angle=380, radius=7mm];

   \draw [->, >=stealth, red, line width=3pt, opacity=0.4, name intersections={of=cut and e-g}] plot [smooth] coordinates {(intersection-1) ($(AS1) + (0,0.2)$) ($(bgp) + (0.0,0.4)$) ($(honeypot-probe-1-1.210) + (0,0.1)$)};
   \draw [-.>, >={Rays[]}, red, line width=3pt, opacity=0.8, name intersections={of=cut and e-g}, shorten >=-2.2mm] plot [smooth] coordinates {(AS9.330) ($(AS7) + (0,0.2)$) ($(AS4) + (0,-0.2)$) (intersection-1)};
   \draw [-.>, >=stealth, red, line width=3pt, dash pattern=on 15pt off 7pt, opacity=0.4] plot [smooth] coordinates {($(honeypot-probe-1-1.230) + (0,-0.1)$) ($(bgp) + (0,-0.2)$) ($(AS2.south)$) (AS0.10)};

   \draw[line width=2pt, black, densely dashed] [use path=\cut];
  \end{pgfonlayer}

 \begin{pgfonlayer}{background}
  \node [draw=gray, fit=(bgp) (honeypot-probe-1-1) (honeypot-probe-n-1) (honeypot-control-1) (label) (net-probe-1) (net-probe-n) (net-control), inner sep=0.25cm, dashed, fill=gray, fill opacity=0.1, rounded corners=0.16cm] (bgpeekaboo) {} ;
    \node [opacity=0.5, yshift=-0.2cm,anchor=west] at (bgpeekaboo.north west) {\toolname};
   \end{pgfonlayer}

 \end{tikzpicture}
 \caption{\toolname overview. After poisoning $E$, the attack towards $V$ is no longer observed at the amplification honeypots and must therefore originate from either $A$, $B$, $C$, or $E$.}\label{fig:system}
\end{figure*}

\section{Background}\label{sec:background}
In this section, we give a brief recap on amplification attacks and BGP Poisoning.

\subsection{Amplification DDoS Attacks}
In an amplification DDoS attack, the attacker tricks public UDP services (e.g., DNS servers) into sending large amounts of traffic to the victim.
This is possible, because IP spoofing enables an attacker to spoof the source address of packets to be the address of the \emph{victim}.
The service will then perceive this packet as a legitimate request and respond to the victim (making the service an involuntary \emph{reflector}).
By carefully selecting reflectors that send large responses, the attacker can maximize the traffic that is reflected to the victim and achieve traffic \emph{amplification}.
IP spoofing further hides the attacker's (network) location, which makes finding the attacker behind an amplification attack notoriously difficult.

Fortunately, amplification attacks have been identified to largely be launched from \emph{single} sources such as Booter services~\cite{kramer2015amppot,krupp2017linking}.
For attackers, finding such reliable and capable source infrastructures is challenging.
Consequently, these infrastructures are usually reused over long time spans.
These services have also been reported to reuse the same set of reflectors for multiple attacks over an extended period of time~\cite{krupp2017linking}.
We can therefore assume that reflectors continuously\footnote{As we will show in \autoref{sec:eval:realworld:activity}, even though individual attacks might be too short for BGP-based traceback, we can aggregate multiple attacks to the same origin, thereby fulfilling this demand.} receive spoofed attack traffic from the same origin.
\subsection{The Border Gateway Protocol (BGP)}
The Internet is often described as a ``network of networks'', as it comprises thousands of so-called \emph{autonomous systems} (ASes).
Every AS is a network under the control of a single entity.
An AS is usually responsible for a number of IP prefixes and can be identified by its unique \emph{AS number} (ASN).

BGP~\cite{rfc4271} enables routing between ASes.
In BGP, ASes exchange routing information with their neighbors through route advertisements (also called announcements).
Each route advertisement describes a path of ASes (\texttt{AS\_PATH}) via which a certain IP prefix may be reached.
When a BGP router receives an advertisement for a prefix, it first checks if it already knows a \emph{better}\footnote{according to its operator defined policies} route for that prefix.
If it does, the new route is only kept as a fallback.
Otherwise, the router prepends its own ASN to the \texttt{AS\_PATH} and advertises this new route to its neighbors.
Between two BGP routers, a new advertisement for a prefix also implicitly withdraws the old route advertised for that prefix.
When routing traffic, packets are forwarded according to the best known route for the most specific prefix covering the traffic's destination.
We assume that the attacker does not control an entire AS, but is a customer of an RFC-compliant AS.

\subsubsection{BGP Poisoning:}\label{sec:background:poisoning}
BGP detects and prevents loops~\cite{rfc4271}.
Before considering new advertisements, routers check if their own ASN is already included in the \texttt{AS\_PATH}.
If so, the new advertisement will be considered as a withdrawal only and no longer propagated to the AS's neighbors.
A side-effect of loop detection is that it enables \emph{BGP Poisoning}.
By crafting the \texttt{AS\_PATH} to include other systems' ASNs, loop detection can intentionally be triggered at these other systems.
Specifically, to trigger loop detection at ASes $X_1, \dots, X_n$, an AS $A$ may send an advertisement with
\[
	\texttt{AS\_PATH} = \left(A, X_1, \dots, X_n, A\right)
\]
If any AS $\in X_1, \dots, X_n$ receives this advertisement, it will find itself already present in the \texttt{AS\_PATH}, consider this a ``loop'', and handle it as a withdrawal subsequently.
The first $A$ ensures that $A$'s neighbors correctly see $A$ as the next on-path AS, while the last $A$ ensures that the prefix is still seen as originating from $A$ (as required, e.g., for Route Origin Validation~\cite{rfc6480,rfc6482}).
We will call $A$ the \emph{poisoning AS} and $X_1$ through $X_n$ the \emph{poisoned ASes}.
Despite its negative name, BGP Poisoning does not imply malevolence---after all, dropping advertisements can only impede reachability of the \emph{poisoning} AS, but not others.
On the contrary, since it gives operators a way to control \emph{in}bound traffic paths, its utility has been proven for many traffic engineering tasks~\cite{katz2011machiavellian,katz2012lifeguard,smith2018routing}.

\section{BGP-based Traceback}\label{sec:traceback}
As noted in \autoref{sec:background:poisoning}, BGP Poisoning can be used to discard route advertisements at other ASes.
In this section we show how the resulting side-effects can be leveraged to find the origin AS of spoofed attack traffic and present our traceback system \toolname.

\subsection{Poisoning for Traceback}\label{sec:traceback:poisoning}
Assume an AS $A$ is sending spoofed traffic to a reflector and receives a poisoning advertisement for the reflector's prefix.
Since the AS will handle this advertisement like a withdrawal, it will remove its routing information for that prefix.
However, without routing information it can no longer send traffic to the reflector.
The reflector will hence stop receiving traffic from $A$.

A similar observation can be made for ASes \emph{forwarding} traffic to the reflector.
Assume an AS $F$ is normally forwarding spoofed traffic from $A$ to the reflector and receives a poisoning advertisement.
Next to losing the ability to send traffic to the reflector, $F$ also has to withdraw any routes for that prefix it had advertised to its neighbors.
Thus $A$ can no longer route traffic to the reflector via $F$.
If no alternative path from $A$ to the reflector exists that circumvents $F$, then $A$ again loses connectivity and traffic at the reflector will stop.
However, even if an alternative path exists, this case might be observable at the reflector:
Unless the IP hop count along both paths is exactly the same, the traffic's TTL value at the reflector will change.

This allows us to check whether some AS $X$ was \emph{on-path} of a traffic flow:
If poisoning $X$ causes the traffic to stop or its TTL value to change, $X$ \emph{was} on-path.
Furthermore, if traffic stopped, the origin AS has no alternative routes avoiding the poisoned AS $X$.

\subsubsection{Default Routes}\label{sec:traceback:defaultroutes}
In practice, some ASes can still route traffic even for prefixes that they have no explicit routing information for.
These \emph{default routes} can be either configured statically or an upstream AS can advertise itself as the next hop for a large prefix (e.g., \texttt{0.0.0.0/0}).
In these cases, the poisoned AS will not loose connectivity entirely, but switch to its default route.
As before, such a route change can result in an observable TTL change.
Furthermore, since a default route cannot be advertised with a more specific prefix than the original route, BGP will avoid paths including a default route whenever possible.

\subsubsection{Combined Probes}\label{sec:traceback:parallel}
The on-path check can also be performed for multiple ASes simultaneously using a combined poisoning advertisement.
If the poisoning of an AS leads to a stop in traffic, poisoning additional ASes cannot undo this effect.
If, however, it leads to a TTL change only, poisoning additional ASes may cause further re-routing or even eliminate alternative paths.
While this could negate the TTL change under specific circumstances\footnote{if it causes traffic to take a path with the exact same IP hop count as the original path}, it results in additional traffic stops in many cases.

Ultimately though, a combined probe can only tell whether \emph{any} of the poisoned ASes were on-path.
To find the exact on-path AS within the probed set, we can use a binary search approach, iteratively splitting the probing set in half and repeating the probing for each half.

\subsubsection{Active Measurements}\label{sec:traceback:active}
This technique can be further supplemented by active measurements.
By provoking replies from a host in the measured AS (e.g. through ICMP Pings or TCP SYNs\footnote{suitable candidates and ports could be found through Internet-wide scanning or by leveraging a search engine such as Shodan\cite{shodan}}), we can observe the effect of poisoning on the AS.
If the replies to our active measurements stop under poisoning, but the spoofed traffic continues (or vice versa), we can conclude that the measured AS was \emph{not} the spoofing source.

\subsubsection{Probes in \toolname}
\toolname uses all but one of its prefixes as \emph{probe prefixes} to probe network responses to poisoning advertisements.
The remaining prefix is designated as the \emph{control} prefix and will only be advertised regularly (non-poisoned).
We will refer to honeypots in probe prefixes and the control prefix as \emph{probe honeypots} and \emph{control honeypots} respectively.

To probe an AS, the BGP router of \toolname sends a poisoned route advertisements for a probe prefix that receives attack traffic.
Once the routes have stabilized, it records the impact on the attack traffic and on pings.
If the attack is also observed by some control honeypots, impact can also be measured by comparing traffic between the probe and control honeypots.

\section{A naive Traceback Approach}\label{sec:traceback:naive}
Using these measures we propose the naive traceback algorithm depicted in \autoref{alg:naive}.
The algorithm simply loops over all ASes (set $\allAS$) in chunks of size at most $n$ and considers each chunk as a combined probe $\probe$.

The actual probing is performed by \textsc{Probe}, which sends out a poisoned advertisement for the ASes in $\probe$, performs the active measurements, and returns the observed effects.
It returns the effect on the attack traffic~($\rp$) and the active measurement results~($\ra$).
$\rp$ can be either $\noeffect$ if no change was observed, $\ttlchange$ if we observed a TTL change, or $\trafficstop$ if traffic has stopped entirely.
Similarly, $\ra$ is a vector with one component (either $\noeffect$ or $\trafficstop$) per probed AS.

If poisoning of $\probe$ shows an effect on the attack traffic, the probe is then narrowed down to find the exact AS(es) that caused this effect.
For this, first, all probed ASes that show an inconsistent behavior in their \emph{active} measurements are discarded from the probe (\textsc{Update}).
If this already reduces the probe to a single consistent AS, it is then added to the candidate set $\candidates$ as a confirmed on-path AS.
Otherwise, the probe is split in half (\textsc{Split}) and the probing is repeated for each half recursively.
Once all ASes have been probed that way, the final set of confirmed on-path ASes $\candidates$ is returned.

\begin{figure}[!t]
  \begin{algorithmic}
    \Procedure{NaiveTraceback}{$\allAS, A, n$}
      \State $\candidates \gets \emptyset$ \Comment{candidates}
      \For{block $\probe$ in $\allAS$ of size $\leq n$}
        \State \Call{ProbeAndUpdate}{$\probe$}
      \EndFor
      \State \Return $\candidates$
    \EndProcedure
    \Statex
    \Procedure{ProbeAndUpdate}{$\probe$}
      \State $\rp, \ra \gets$ \Call{Probe}{$\probe$}
      \If{$\rp \neq \noeffect$}
	\State $\probe \gets$ \Call{Update}{$\probe, \rp, \ra$}
	\If{$\left|\probe\right| = 1$}
	  \State $\candidates \gets \candidates \cup \probe$ \Comment{AS in $\probe$ was on-path}
	\ElsIf{$\left|\probe\right| \geq 2$}
	  \State $\probe_1, \probe_2 \gets$ \Call{Split}{$\probe$} \Comment{``Binary Search''}
	  \State \Call{ProbeAndUpdate}{$\probe_1$}
	  \State \Call{ProbeAndUpdate}{$\probe_2$}
	\EndIf
      \EndIf
    \EndProcedure
    \Statex
    \Procedure{Update}{$\probe, \rp, \ra$}
      \If{$\rp = \trafficstop$}
        \State $\exclude \gets \left\{X \in \probe\mid\ra[X] \neq \trafficstop\right\}$
      \Else
        \State $\exclude \gets \left\{X \in \probe\mid\ra[X] = \trafficstop\right\}$
      \EndIf
       \State \Return $\probe \setminus \exclude$
    \EndProcedure
 \end{algorithmic}
 \caption{Naive traceback algorithm}\label{alg:naive}
\end{figure}

\subsection{Runtime Analysis}\label{sec:traceback:naive:runtime}
The runtime of this traceback approach is greatly dominated by the probing step, which involves sending out a poisoned advertisement and performing active measurements.
This is because (1) it may take several minutes for routes to settle after a new advertisement~\cite{labovitz2000delayed,katz2012lifeguard}, only after which active measurements can be performed, and (2) several BGP mechanisms further limit the rate at which routers may send out new advertisements (\autoref{sec:discussion:bgp}).
Therefore, realistically, advertisements cannot be made much faster than once every ten minutes.

We will thus count the number of required probing steps to analyze the traceback runtime.
At chunk size $n$ and a total of $N$ ASNs in $\allAS$, the naive traceback algorithm from \autoref{alg:naive} takes $\left\lceil\frac{N}{n}\right\rceil$ steps to test each chunk once, plus an additional $2\log_2 n$ steps for every on-path AS to reduce the chunk it is contained in down to a single AS.
For example, the AS65000 BGP Routing Table Analysis Report~\cite{potaroo} lists roughly $66000$ active ASes for the end of 2019 and an average AS path length of $5.5$.
With a chunk size of $n=128$ this thus gives an average of $593$ steps total, which, at 6 advertisements per hour, would take about 4 days and 3 hours to complete.

Given that spoofing sources, such as Booter services, are active for extended periods of time (see \autoref{sec:eval:realworld:activity}) and reuse the same amplifiers for a week or longer~\cite{krupp2017linking}, this shows that BGP-based traceback is feasible in principle.

\subsection{Discussion}

While this algorithm is intuitive and requires no external knowledge of AS properties or relationships, it comes with two main drawbacks:
\begin{enumerate*}[label=(\arabic*)]
 \item It only returns an unordered set of on-path ASes, which still leaves the exact path and origin unknown.
 \item It ``wastes'' a lot of time poisoning off-path ASes that could potentially be avoided, as it effectively conducts an exhaustive search over the entire AS space.
\end{enumerate*}

\subsubsection{On-Path Ordering}
Ideally, one would like to find the true origin AS of the spoofed traffic, or at least the on-path AS closest to it that can still be discovered.
However, through probing we can only tell \emph{whether} an AS was on-path or not, but not its position along the path.

While at first glance it seems that this problem could be solved by comparing TTL values received from hosts located in these ASes, this is not necessarily true:
Although the AS level path should be the same for both traffic originating from and traffic forwarded by an AS, the IP level paths (and hence their hop counts) can differ vastly.
In a similar vein, one might attempt to infer the on-path order from traceroutes towards hosts in the candidate ASes, mapping the obtained IP level traceroute paths into AS level paths.
Barring the problems of mapping IP to AS level paths~\cite{zhang2010quantifying, hyun2003traceroute}, traceroutes from the vantage point can only reveal paths \emph{towards} other hosts, but not their reverse paths, which we are interested in.

\subsubsection{Runtime Improvement}
Although the estimated runtime does not seem prohibitive, the question still remains whether additional knowledge about ASes, such as their relationships, can be used to achieve effective traceback more efficiently.
For example, as a simple optimization the search can be aborted as soon as a stub AS is confirmed to be on-path.
Since stub ASes do not provide transit for other ASes, an on-path stub AS must be the one originating the observed traffic.
In these cases, such an early termination will reduce the expected runtime to $1/2$ of the original algorithm.
In the next section, we show how the runtime can be further improved by leveraging AS relationship data as well as the information about alternative path availability one can obtain from probing.

\section{Flow Graph based Traceback}\label{sec:graph}
Our graph-based traceback algorithm exploits knowledge about the relationship between ASes to limit the search space.
For example, if the attack traffic stops under poisoning, we know that the source has no alternative route that avoids the poisoned ASes.
To efficiently reason about alternative paths and whether an advertisement might be propagated from one AS to another, we define a so-called \emph{AS Flow Graph}.

\subsection{AS Flow Graphs}\label{sec:graph:flowgraph}
AS relationships are usually classified as either \rCPv (\rCP) or \rPPv (\rPP) relations~\cite{gao2001inferring}.
In a \rCPv relation, one AS (the customer) pays another (the provider) for transit such that the customer may reach and be reached from the Internet via the provider.
In a \rPPv relation, two ASes agree to transit traffic for their customers to one another, thereby reducing the amount of traffic they would have to pay their provider for otherwise.

The resulting BGP paths are generally assumed to be valley-free~\cite{gao2001inferring}: zero or more \rCPv edges (``up''), followed by at most one \rPPv edge (``sideways''), followed by zero or more \rPCv edges (``down'').
In other words: a \rPPv or \rPCv edge can never be followed by a \rCPv or \rPPv edge, as this would result in a ``valley''.
Note that this property is symmetric and holds for both, the AS-level paths taken by routed traffic as well as the propagation paths of BGP advertisements.

Following this, an AS may receive an advertisement in two states:
If the advertisement was received from a customer (i.e., via a \rCPv edge), the valley-free assumption does not restrict the edge types that may follow.
We will call this state \unconstrained.
If an advertisement was received from either a peer or a provider, it may only be forwarded to customers, but not to other peers or providers.
We will thus call this state \constrained.

We can use a graph to model advertisement propagation that uses \emph{two} nodes per AS, one for each state.
Formally, we define this graph $G = (V, E)$ as follows:
Every AS $A$ is represented by two vertices, $u_A$ and $c_A$, representing the \unconstrained and \constrained state respectively,
\[
	V = \bigcup\limits_{A \in \AS} \left\{ u_A, c_A \right\}
\]
We will denote the AS represented by a node x using $\asn(x)$,
\[
 \asn(x) = A \Leftrightarrow x \in \left\{u_A, c_A\right\}
\]
When AS $A$ is a provider of AS $B$, advertisements may only be forwarded ``downhill'' from $c_A$ to $c_B$ or ``uphill'' from $u_B$ to $u_A$.
If $A$ and $B$ share a \rPP relation, then advertisements may only be forwarded between $A$ and $B$ at the ``peak'' of the path, i.e., from $u_A$ to $c_B$ or from $u_B$ to $c_A$.
Finally, advertisements received at $u_A$ may of course also be forwarded to $c_A$.
In total,
\begin{align*}
	E = &\left(\bigcup\limits_{A,B \in \CP} \left\{ (u_A, u_B), (c_B, c_A)\right\}\right)\\
	&\cup \left(\bigcup\limits_{A,B \in \PP} \left\{ (u_A, c_B), (u_B, c_A) \right\} \right)\\
	&\cup \left(\bigcup\limits_{A \in \AS} \left\{(u_A, c_A)\right\} \right)
\end{align*}

This graph then captures all valley-free AS paths, and every path in this graph corresponds to a valid valley-free AS path.
For a formal proof, please refer to \autoref{app:graph:proof}.

By construction, the orientation of edges denotes the direction of advertisement propagation.
For example, an edge from $u_B$ to $u_A$ implies that advertisements received by $B$ may be further propagated to $A$.
However, the inverse graph obtained by reversing all edges is meaningful as well, as it describes all possible traffic flows between ASes:
If AS $B$ advertises a route for a prefix to AS $A$, then $A$ may send traffic destined towards that prefix to $B$.

In the example given in \autoref{fig:flowgraph}, $A$ is a provider of $B$, $C$ is a provider of $D$, and $B$ and $C$ have a peer-to-peer relation.
However, it is not obvious from the relationship graph, whether advertisements from $D$ may reach $A$.
The resulting flowgraph answers this unequivocally:
Advertisements from $D$ can reach $C$ ($u_D \rightarrow u_C$) and $B$ ($u_D \rightarrow u_C \rightarrow c_B$), but cannot reach $A$, since there is no directed path from either $u_D$ or $c_D$ to $u_A$ or $c_A$.

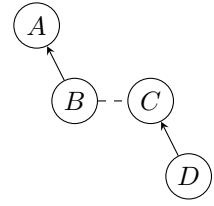
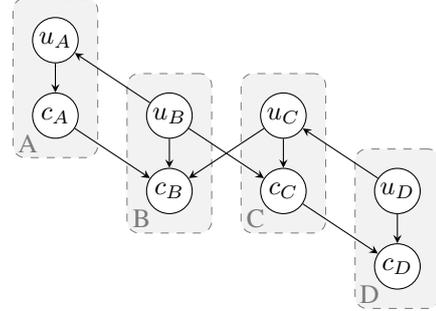
\begin{figure}[t]
 \centering
 \begin{subfigure}[t]{\columnwidth}
 \centering
 \begin{tikzpicture}[]

  \node [circle,draw,minimum size=6mm,label=center:$A$] (a) at (-0.5,0) {};
  \node [circle,draw,minimum size=6mm,label=center:$B$] (b) at (0,-1) {};
  \node [circle,draw,minimum size=6mm,label=center:$C$] (c) at (1,-1) {};
  \node [circle,draw,minimum size=6mm,label=center:$D$] (d) at (1.5,-2) {};

  \draw[->,>=stealth] (b) -- (a);
  \draw[dashed] (b) -- (c);
  \draw[->,>=stealth] (d) -- (c);
 \end{tikzpicture}
 \caption{Example AS relationships: $A$ is provider of $B$, $C$ is provider of $D$, and $B$ and $C$ have a peer-to-peer relation.}
 \end{subfigure}

 \begin{subfigure}[t]{\columnwidth}
 \centering
 \begin{tikzpicture}[]

  \begin{pgfonlayer}{main}
  \node [circle,draw,minimum size=6mm,label=center:$u_A$, fill=white] (u-a) at (-1.5,0.5) {};
  \node [circle,draw,minimum size=6mm,label=center:$c_A$, fill=white] (c-a) at (-1.5,-0.5) {};
  \node [circle,draw,minimum size=6mm,label=center:$u_B$, fill=white] (u-b) at (0,-0.5) {};
  \node [circle,draw,minimum size=6mm,label=center:$c_B$, fill=white] (c-b) at (0,-1.5) {};
  \node [circle,draw,minimum size=6mm,label=center:$u_C$, fill=white] (u-c) at (1.5,-0.5) {};
  \node [circle,draw,minimum size=6mm,label=center:$c_C$, fill=white] (c-c) at (1.5,-1.5) {};
  \node [circle,draw,minimum size=6mm,label=center:$u_D$, fill=white] (u-d) at (3,-1.5) {};
  \node [circle,draw,minimum size=6mm,label=center:$c_D$, fill=white] (c-d) at (3,-2.5) {};

  \draw[->,>=stealth] (u-a) -- (c-a);
  \draw[->,>=stealth] (u-b) -- (c-b);
  \draw[->,>=stealth] (u-c) -- (c-c);
  \draw[->,>=stealth] (u-d) -- (c-d);

  \draw[->,>=stealth] (u-b) -- (u-a);
  \draw[->,>=stealth] (c-a) -- (c-b);

  \draw[->,>=stealth] (u-d) -- (u-c);
  \draw[->,>=stealth] (c-c) -- (c-d);

  \draw[->,>=stealth] (u-b) -- (c-c);
  \draw[->,>=stealth] (u-c) -- (c-b);
  \end{pgfonlayer}

   \begin{pgfonlayer}{background}
    \node [draw=gray, fit=(u-a) (c-a), inner sep=0.25cm, dashed, fill=gray, fill opacity=0.1, rounded corners=0.16cm] (a) {} ;
    \node [opacity=0.5, xshift=0.2cm, yshift=0.2cm] at (a.south west) {A};
    \node [draw=gray, fit=(u-b) (c-b), inner sep=0.25cm, dashed, fill=gray, fill opacity=0.1, rounded corners=0.16cm] (b) {} ;
    \node [opacity=0.5, xshift=0.2cm, yshift=0.2cm] at (b.south west) {B};
    \node [draw=gray, fit=(u-c) (c-c), inner sep=0.25cm, dashed, fill=gray, fill opacity=0.1, rounded corners=0.16cm] (c) {} ;
    \node [opacity=0.5, xshift=0.2cm, yshift=0.2cm] at (c.south west) {C};
    \node [draw=gray, fit=(u-d) (c-d), inner sep=0.25cm, dashed, fill=gray, fill opacity=0.1, rounded corners=0.16cm] (d) {} ;
    \node [opacity=0.5, xshift=0.2cm, yshift=0.2cm] at (d.south west) {D};
   \end{pgfonlayer}
 \end{tikzpicture}
 \caption{The resulting AS flowgraph shows that route advertisements from $D$ may reach $C$ and $B$, but can never reach $A$.}
 \end{subfigure}
 \caption{AS flowgraph example}\label{fig:flowgraph}
\end{figure}

\subsubsection{Reachability and Dominance}
A poisoned advertisement can affect ASes in two ways:
ASes included in the advertisements' \texttt{AS\_PATH} are affected directly, as they will discard the advertisement due to loop detection.
However, this also prevents them from propagating the advertisement further, which, in turn, can affect other ASes.
While some of these indirectly affected ASes may still receive the advertisement via alternative routes and thus only experience a route change, others may no longer be able to receive the advertisement at all.
To reason about the propagation of (poisoned) advertisements originating from a specific AS $A$, we can define its AS-specific subgraph:
\begin{definition}[AS-specific subgraph]
 For the flow graph $G$ and AS $A$, we define the AS-specific rooted subgraph $G_A = (V_A, E_A, u_A)$ as the subgraph of $G$ rooted at $u_A$.
\end{definition}
We can then define two relations on this graph, reachability and (joint) dominance, to capture which ASes \emph{might} potentially be affected by a poisoned advertisement and which \emph{will} be inevitably.

Consider an advertisement that is poisoning ASes ${P = \{X_1, \dots, X_n\}}$ (with corresponding nodes ${p = \{u_{X_1}, c_{X_1}, \dots, u_{X_n}, c_{X_n}\}}$).
Another AS $Y$ might only be affected by this advertisement if it can receive advertisements from $A$ via any AS in $P$\footnote{Under specific circumstances, poisoned advertisements may also induce route changes at ASes that have no connection to a poisoned ASes. We analyze these cases in \autoref{sec:discussion:induced}.}.
As all paths in the graph $G_A$ describe valid propagation paths for advertisements originating from $A$, AS $Y$ therefore might be affected if there is a path from any node $x \in p$ to any node $y \in \{u_Y, c_Y\}$.
\begin{definition}[Reachability]
	We call a node $y$ \emph{reachable} from another node $x$, iff there exists a path in $G_A$ from $x$ to $y$.
	\begin{align*}
		&y \in \reachable_{G_A}(x) \Leftrightarrow\\
		&\quad\exists \nodepath = (x_1 = x, \dots, x_n = y) :\\
		&\quad\quad\forall 1 \leq i < n : (x_i, x_{i+1}) \in E_A
	\end{align*}
\end{definition}

The set $\reachable_{G_A}(x)$ describes all nodes whose traffic towards $A$ \emph{may} be routed via $x$.
The definition of $\reachable$ can be trivially extended to sets of nodes by taking their union:
\[
	\reachable_{G_A}(\{x_1, \dots, x_n\}) = \bigcup\limits_{i=1}^n \reachable_{G_A}(x_i)
\]
Therefore, the set of ASes that \emph{might} be affected by an advertisement poisoning nodes $p$ is simply $\reachable_{G_A}(p)$.
That is, if poisoning of nodes $p$ causes the TTL values of the attack traffic to change we can infer that the origin was affected and must be an element of $\reachable_{G_A}(p)$.

In a similar vein we can also define nodes that \emph{must} be affected.
Intuitively, a node $y$ must be affected if it cannot receive advertisements from $A$ but via nodes in $p$.
Thus, once all nodes in $p$ are removed from the graph, $y$  should be no longer reachable from $u_A$.
In graph theory terms $p$ then constitutes a \emph{$u_A$-$y$-vertex-separator} (albeit defined over directed graphs), but can also be seen as a generalization of the concept of \emph{dominators} from control-flow-graph analysis.
\begin{definition}[(Joint) Domination]
	We call a node $y$ (jointly) \emph{dominated} by a set of nodes $\{x_1, \dots, x_n\}$ in graph $G_A$, iff $y$ is reachable from $u_A$ but removing the set $\{x_1, \dots, x_n\}$ breaks reachability from $u_A$ to $y$.
	Formally
	\begin{align*}
		&y \in \dominatees_{G_A}(\{x_1, \dots, x_n\}) \Leftrightarrow\\
		&\quad y\in \reachable_{G_A}(u_A) \land y \not\in \reachable_{G'_A}(u_A)
	\end{align*}
	where $G'_A = (V'_A, E'_A, u_A)$ with ${V'_A = V_A \setminus \{x_1, \dots, x_n\}}$ and ${E'_A = \left\{(x, y) \in E_A \mid x \in V'_A \land y \in V'_A \right\}}$.
\end{definition}

Hence, the set of ASes that \emph{will} inevitably be affected, i.e., those that will have to switch to default routes or possibly loose connection, by an advertisement poisoning nodes $p$ is $\dominatees_{G_A}(p)$.
As noted above, other ASes from $\reachable_{G_A}(p)$ might be affected as well, but for those contained in $\dominatees_{G_A}(p)$ we can be certain.

\subsection{Flow Graph based Traceback}\label{sec:graph:traceback}
We can leverage this flow graph for AS traceback in two ways:
First, given a known on-path AS we know that the neighboring on-path AS must also be one of its successors in the graph.
Instead of globally searching for the origin AS we can therefore iteratively find all on-path ASes one-by-one by probing the successors of the latest found on-path AS.

Second, we can use reachability and dominance relations given by the graph to reduce our probing search space:
If traffic stops, we know that the origin AS has no alternative paths available and can thus limit our search space to nodes dominated by the probed ASes (which includes the probed ASes themselves).
If we observe a TTL change, we can still infer that the origin AS is at least indirectly affected by the probe and can therefore limit our search space to those nodes that are reachable from the probed ASes.
We can also use similar inferences on the results of our active measurements:
When responses to active measurements for a probed AS stop but the attack traffic does not, we can exclude all nodes dominated by the probed AS from our search space, as those would no longer be able to send traffic as well.
Vice versa, if the attack traffic stops but active measurements show that a probed AS is still replying to pings, we can exclude all nodes reachable by the probed AS, as they also could send traffic to us via the probed AS.

\subsubsection{Graph based Traceback Algorithm}\label{sec:graph:traceback:algorithm}
Combining both of these effects leads to the graph-based traceback algorithm shown in \autoref{alg:traceback}.
Given the rooted subgraph ${G_A = (V_A, E_A, u_a)}$ of the receiving AS $A$ and a probe size limit $n$, it returns a set of candidate source ASes.

\begin{figure}[!t]
 \begin{algorithmic}
    \Procedure{Traceback}{$G_A, n$}
      \State $\candidates \gets \left\{\asn(v)\mid v \in V_A\right\})$ \Comment{candidates}
      \State $\logbook \gets \emptyset$ \Comment{logbook}
      \State $\last \gets A$ \Comment{most recent on-path AS}
      \While{$\candidates \setminus \logbook \neq \emptyset$}
        \State $\probe \gets$ \Call{PickProbe}{$\candidates, \logbook,  \last, n$}
	\State $\logbook \gets \logbook \cup \probe$
        \State \Call{ProbeAndUpdate}{$\probe$}
      \EndWhile
      \State \Return $\candidates$
    \EndProcedure
    \Statex
    \Procedure{ProbeAndUpdate}{$\probe$}
      \State $\rp, \ra \gets$ \Call{Probe}{$\probe$}
      \State \Call{UpdatePassive}{$\probe, \rp$}
      \State \Call{UpdateActive}{$\probe, \rp, \ra$}
      \If{$\rp \neq \noeffect$}
	\If{$\left|\probe\right| = 1$}
	  \State $\last \gets X \in \probe$ \Comment{AS $X$ was on-path}
	\ElsIf{$\left|\probe\right| \geq 2$}
	  \State $\probe_1, \probe_2 \gets$ \Call{Split}{$\probe$} \Comment{``Binary Search''}
	  \State \Call{ProbeAndUpdate}{$\probe_1$}
	  \State \Call{ProbeAndUpdate}{$\probe_2$}
	\EndIf
      \EndIf
    \EndProcedure
    \Statex
    \Procedure{UpdatePassive}{$\probe, \rp$}
      \If{$\rp = \trafficstop$}
        \State $\candidates \gets \candidates \cap \dominatees_{G_A}(\probe)$
      \ElsIf{$\rp = \ttlchange$}
          \State $\candidates \gets \candidates \cap \reachable_{G_A}(\probe)$
      \EndIf
    \EndProcedure
    \Statex
    \Procedure{UpdateActive}{$\probe, \rp, \ra$}
      \If{$\rp = \trafficstop$}
        \State $\exclude \gets \left\{X \in \probe\mid\ra[X] \neq \trafficstop\right\}$
        \State $\candidates \gets \candidates \setminus \reachable{G_A}(\exclude)$
      \Else
        \State $\exclude \gets \left\{X \in \probe\mid\ra[X] = \trafficstop\right\}$
        \State $\candidates \gets \candidates \setminus \dominatees_{G_A}(\exclude)$
      \EndIf
      \State $\probe \gets \probe \setminus \exclude$
    \EndProcedure
 \end{algorithmic}
 \caption{Flow Graph based traceback algorithm}\label{alg:traceback}
\end{figure}

For this, the algorithm maintains a set of source candidates $\candidates$, which is initially set to all ASes reachable from $A$ and then continuously narrowed down during traceback.
Further, it also keeps a logbook $\logbook$ of ASes that have already been probed and can thus be excluded from further probing, as well as the most recent on-path AS $\last$, which is initially set to $A$.
As long as there are candidates that have not been probed yet ($\candidates \setminus \logbook$), a new set of ASes to probe is selected (see \autoref{sec:graph:traceback:selection}) and passed to \textsc{ProbeAndUpdate}, which will perform the probing and update the candidate set and logbook accordingly.

As before, probing is performed by \textsc{Probe} (see \autoref{alg:naive}).
Updating the candidate set is performed in two steps:
First, \textsc{UpdatePassive} reduces the candidate set based on the overall effect on the attack traffic, limiting $\candidates$ to the set of nodes dominated or reachable by the current probe if traffic stopped or a TTL change was observed.
Second, \textsc{UpdateActive} further narrows down the candidate set and the current probe by finding probed ASes whose active measurements are inconsistent with the overall observation and then removing nodes dominated or reachable by these.

If there has been an effect on the attack traffic and multiple probed ASes exhibit a consistent behavior, they are split in half (\textsc{Split}) and the probing is repeated for each half.
If a probe has been narrowed down to a single consistent AS, this AS is set as the most recent on-path AS $\last$ and a new probe is selected.

\subsubsection{Probe Selection}\label{sec:graph:traceback:selection}
Probe Selection (\textsc{PickProbe}) picks a new probe based on the current candidates $\candidates$, the ASes that have been probed before $\logbook$, the last discovered on-path AS $\last$, and the maximum probe size $n$.
As discussed above, the next on-path AS must be a successor of $\last$.
We can thus partition our remaining search space of unprobed candidates $\candidates \setminus \logbook$ into ``layers'' according to their distance to the most recent on-path AS $\last$.
The next on-path AS should then be part of the nearest layer.
This ordering also ensures that, should we be unable to find the direct on-path successor of $\last$, e.g., because it exhibits no observable poisoning reaction, we only gradually expand our search scope to later on-path ASes.

Intuitively, to reduce the candidate set as fast as possible, we would like to first probe ASes that have a high impact on the candidate set.
While a stub AS could also have a high impact if it is on-path (effectively reducing the candidate set to a single entry), statistically speaking, in most cases probed ASes will be off-path ASes.
The majority of candidate set reductions will thus occur in \textsc{UpdateActive}.
We therefore rank ASes by the number of candidates that are \emph{reachable} through them, i.e., $\left|\reachable_{G_A}(X) \cap \candidates\right|$, and then select the $n$ highest ranking ASes as the next probe.

\section{Evaluation}\label{sec:evaluation}
We next turn to evaluate the runtime and success rate of \toolname and compare both proposed algorithms.
For this, we were fortunate enough to obtain a temporary ASN and a temporary /22 prefix allocations for research purposes from our regional Internet registry and were granted access to the PEERING BGP testbed~\cite{schlinker2019peering}.
However, PEERING's path length limit unfortunately proved prohibitive for any actual traceback runs of \toolname on the Internet:
Although a recent study showed that advertisements even with long paths of up to 255 hops are propagated to the vast majority of the Internet~\cite{smith2020withdrawing}, we were only able to send advertisements with up to 5 hops via PEERING.
Since the first and last of these also had to be set to our own temporary ASN, this left us with an effective probe size of 3.
Unfortunately, with this even a single run of our baseline approach would have taken over 157 days to complete.
We therefore resort to simulation for a comparative evaluation of our proposed traceback approaches and use the PEERING testbed for supporting experiments.

\subsection{Simulation Methodology}
As noted by previous works~\cite{schuchard2010losing, schuchard2012routing, smith2018routing}, a complete and fully accurate model of the entire Internet cannot be obtained, as it would require exact knowledge of peering agreements and router configurations, both of which are usually regarded as trade-secrets and thus generally non-public.
Simulation can therefore only be performed over approximate topology data, such as the one regularly published by CAIDA~\cite{caidaASrel}, and by making assumptions on router configurations.

Although the simulator by Smith et al.~\cite{smith2018routing} is thankfully publicly available, we found it unsuitable for our use-case as it does model neither default routes nor TTL values along paths.
Consequently, we designed our own lightweight simulator.

\subsubsection{Simulator Design}
Our simulator is based around the same AS flow graph described in \autoref{sec:graph:flowgraph}.
To model AS paths taken from a source to a destination, our simulator assigns every edge in the graph a random weight between 0 and 100 and then computes the shortest path.
This is in line with the regular BGP decision process~\cite[sec.~9.1]{rfc4271}, which generally prefers shorter AS paths.
The random edge weight hereby models the local preference value a network operator may set to prefer one link over another.
As the graph has \emph{two} nodes corresponding to each AS $X$, $u_X$ and $c_X$, we ensure that the edge $(u_X, c_X)$ is assigned the weight $0$, which also ensures that the resulting AS paths are loop free---a path containing both $u_X$ and $c_X$ can never be shorter than the one that uses the zero-weight $(u_X, c_X)$ edge.

The effect of a poisoning advertisement can then be simulated by temporarily removing the poisoned AS's nodes from the graph, such that they can no longer send traffic to the destination nor forward advertisements to their customers or peers.

In contrast to previous works, our simulator also attempts to model the presence of default routes, which were identified as a major culprit for discrepancies between simulations and experimental results in other BGP-based tools~\cite{smith2020withdrawing}.
For this, every AS is randomly marked as either \emph{having} a default route or not with certain probability.
When simulating a poisoning advertisement, those ASes with default routes are not removed from the graph, but instead increase the weights of their incoming edges by $10000$.
This ensures they are no longer selected as shortest paths, unless no alternative is available---as would be the effect of less-specific prefixes.
ASes with default routes also have a chance of using a secondary set of incoming edge weights when being poisoned, as otherwise their default route would always be identical to their regular route.

Finally, our simulator allows mapping AS paths to IP hop counts.
For this, every step along the AS path is assigned a random hop count value drawn from a negative binomial distribution, which we found to be a good fit after analyzing traceroute results from RIPE Atlas~\cite{ripeatlas} (see \autoref{sec:eval:realworld:ttl}).
In the real Internet, the IP-level path length also depends on which ingress and egress routers are taken.
We therefore make this random value also dependent on the preceding and succeeding AS hop.

\subsection{Results}\label{sec:eval:sim:results}
With this simulator, we conducted multiple experiments to compare our different traceback algorithms in terms of efficacy and efficiency as well as the influence of various parameters, such as the placement of the deployment location, the presence of default routes, and the choice of probe size.

As a baseline setup for our evaluation we placed the deployment location as a customer of the PEERING testbed (AS47065), which fosters comparability with our real-world experiments (\autoref{sec:eval:real}).
We set the default route probability to $40\%$ as a conservative approximation, given that Smith et al.~\cite{smith2020withdrawing} report a default route prevalence between $26.8\%$ and $36.7\%$ in their experiments.
Lastly, we picked $128$ as a conservative probe size, since they also report successful propagation of paths of length up to $255$~\cite{smith2020withdrawing}.
The flow graph used by our simulator was based on the public CAIDA AS relationship dataset for December 2019~\cite{caidaASrel}, which also covered the timeframe when our real-world experiments were performed, augmented by the peering relations listed by the PEERING testbed~\cite{schlinker2019peering}.

Every experiment was repeated $1024$ times, each time with a randomly chosen AS as the traffic's origin.
In addition, the simulator was also given a fresh random seed for every run, such that our results are not biased due to a single lucky weight assignment or similar.
We use \emph{naive} to refer to the naive algorithm, \emph{naive+} for the variant with early termination, and \emph{graph} for the graph-based algorithm.

\begin{table}[t]
 \setlength{\tabcolsep}{4pt}
	\caption{Success rate with a final candidate set of size at most $x$}
	\label{tab:eval:size}
 \centering
 \begin{tabular}{lrrrrrrrr}
$x$ & $\leq 1$ & $\leq 2$ & $\leq 3$ & $\leq 4$ & $\leq 5$ & $\leq 6$ & $\leq 7$ & $\leq 8$\\
\midrule
\multirow{3}{*}{}naive & $\begin{aligned} {\bf 10\%}& \\[-5pt] {\scriptsize \pm 2\%}&\end{aligned}$ & $\begin{aligned} {\bf 32\%}& \\[-5pt] {\scriptsize \pm 3\%}&\end{aligned}$ & $\begin{aligned} {\bf 52\%}& \\[-5pt] {\scriptsize \pm 3\%}&\end{aligned}$ & $\begin{aligned} {\bf 59\%}& \\[-5pt] {\scriptsize \pm 3\%}&\end{aligned}$ & $\begin{aligned} {\bf 60\%}& \\[-5pt] {\scriptsize \pm 3\%}&\end{aligned}$ & $\begin{aligned} {\bf 61\%}& \\[-5pt] {\scriptsize \pm 3\%}&\end{aligned}$ & $\begin{aligned} {\bf 61\%}& \\[-5pt] {\scriptsize \pm 3\%}&\end{aligned}$ & $\begin{aligned} {\bf 61\%}& \\[-5pt] {\scriptsize \pm 3\%}&\end{aligned}$\\[5pt]
naive+ & $\begin{aligned} {\bf 28\%}& \\[-5pt] {\scriptsize \pm 3\%}&\end{aligned}$ & $\begin{aligned} {\bf 46\%}& \\[-5pt] {\scriptsize \pm 3\%}&\end{aligned}$ & $\begin{aligned} {\bf 57\%}& \\[-5pt] {\scriptsize \pm 3\%}&\end{aligned}$ & $\begin{aligned} {\bf 60\%}& \\[-5pt] {\scriptsize \pm 3\%}&\end{aligned}$ & $\begin{aligned} {\bf 61\%}& \\[-5pt] {\scriptsize \pm 3\%}&\end{aligned}$ & $\begin{aligned} {\bf 61\%}& \\[-5pt] {\scriptsize \pm 3\%}&\end{aligned}$ & $\begin{aligned} {\bf 61\%}& \\[-5pt] {\scriptsize \pm 3\%}&\end{aligned}$ & $\begin{aligned} {\bf 61\%}& \\[-5pt] {\scriptsize \pm 3\%}&\end{aligned}$\\[5pt]
graph & $\begin{aligned} {\bf 58\%}& \\[-5pt] {\scriptsize \pm 3\%}&\end{aligned}$ & $\begin{aligned} {\bf 62\%}& \\[-5pt] {\scriptsize \pm 3\%}&\end{aligned}$ & $\begin{aligned} {\bf 65\%}& \\[-5pt] {\scriptsize \pm 3\%}&\end{aligned}$ & $\begin{aligned} {\bf 66\%}& \\[-5pt] {\scriptsize \pm 3\%}&\end{aligned}$ & $\begin{aligned} {\bf 66\%}& \\[-5pt] {\scriptsize \pm 3\%}&\end{aligned}$ & $\begin{aligned} {\bf 67\%}& \\[-5pt] {\scriptsize \pm 3\%}&\end{aligned}$ & $\begin{aligned} {\bf 68\%}& \\[-5pt] {\scriptsize \pm 3\%}&\end{aligned}$ & $\begin{aligned} {\bf 68\%}& \\[-5pt] {\scriptsize \pm 3\%}&\end{aligned}$\\[5pt]
\end{tabular}
\end{table}

\subsubsection{Traceback Success}
To analyze the efficacy of our proposed traceback, we analyzed how often \toolname succeeds in finding the origin.
The naive algorithm starts with an empty candidate set and selectively adds on-path ASes.
Hence we consider it successful if the true origin is contained in its final candidate set.
In contrast, the graph-based algorithm assumes all ASes as candidates initially, but excludes ASes during the traceback run.
Therefore, to be successful, the final candidate set must also be smaller than a certain threshold (ideally of size~1).

As shown in \autoref{tab:eval:size}, we find that both algorithms have similar success rates:
The naive algorithm manages to identify the true origin in $61\% \pm 3\%$\footnote{We denote the $95\%$ Agresti-Coull confidence interval} cases, while the graph-based traceback algorithm achieves a slightly higher success rate of $68\% \pm 3\%$ with a candidate set of 8 or less.
When limiting the graph-based algorithm to a single candidate, it still manages to succeed in $58\% \pm 3\%$ cases.
The naive algorithm is expected to find all on-path ASes and thus has an expected candidate set size of the average path length.
For the graph algorithm, we will use a candidate set size of 8 in the following.

\subsubsection{Runtime}
As noted before, the traceback speed of \toolname is limited by the rate at which BGP advertisements can be sent.
We therefore measure the runtime of each algorithm in the number of probing steps, with a realistic step duration of 10 minutes.
As can be seen from \autoref{fig:eval:steps}, the naive algorithm performs similar to the number of steps derived in \autoref{sec:traceback:naive:runtime}, with a mean and median runtime of $549$ steps.
Adding early termination on stub-ASes slightly reduces the median to $523$ steps, but more importantly reduces the average to just above $400$ steps, with $25\%$ even terminating in under $268$ steps.
Since it is strictly superior to the naive algorithm, we will only report results for the naive+ algorithm in the following.

Overall, the graph-based algorithm performs best by far, with an average of $159$ steps and a median of only $98.5$ steps.
Interestingly, a quarter of all cases terminate in less than $29$ steps, around $4.83$~hours at six advertisements per hour.
Most notably, even in the worst case the graph-based algorithm still completes its search in $415$ steps or less, thus making it faster than even the best run of the naive algorithm.
This demonstrates that utilizing AS graph information substantially improves the efficiency, with a median runtime speed-up of $5.6\times$.
\begin{figure}[t]
	\centering
	\includegraphics[width=0.99\linewidth]{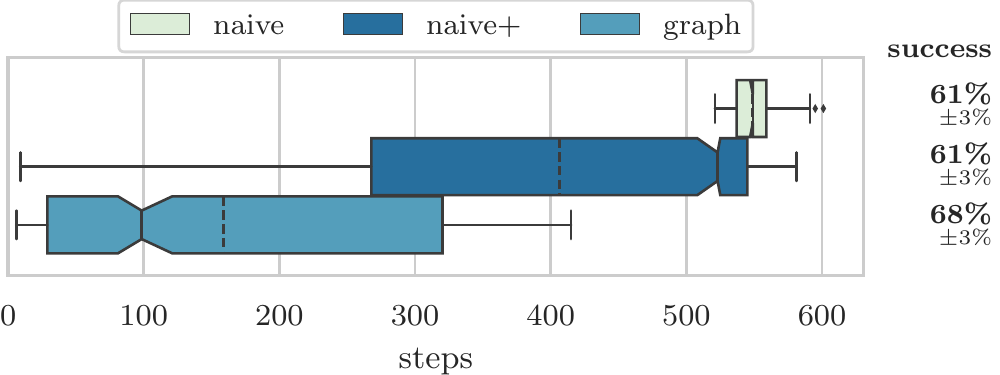}
	\caption{Runtime comparison. Dashed lines show average times, while the notch indicates the $95\%$ confidence interval around the median.}
	\label{fig:eval:steps}
\end{figure}

\subsubsection{Prefix Parallelization}\label{sec:eval:sim:parallel}
If \toolname observes an attack in multiple probing prefixes simultaneously, we can further speed up traceback by running multiple probes in parallel.
While this does not necessarily reduce the total number of \emph{advertisements}, it greatly reduces the number of \emph{steps} and hence the total runtime.
\autoref{fig:eval:prefixes:steps} therefore compares using different numbers of probing prefixes, one (no parallelization), two, and eight.
For the naive algorithm we see a linear speed-up, reducing the median runtime from $523$ to $262$ steps when using two prefixes and $66$ steps when using eight.
With the exception of the "binary search" step, probe selection is independent from previous probe results, and thus the naive algorithm can be almost perfectly parallelized.
The graph-based algorithm benefits less from parallelization, since probe selection is strongly coupled to previous probe results, but still sees a $5.8\times$ speed-up in the median runtime from $98.5$ steps to $17$ steps when going to eight prefixes.
As with the probe size, varying the number of prefixes does not have a significant influence on the success rate.

\begin{figure}[t]
	\centering
	\includegraphics[width=0.99\linewidth]{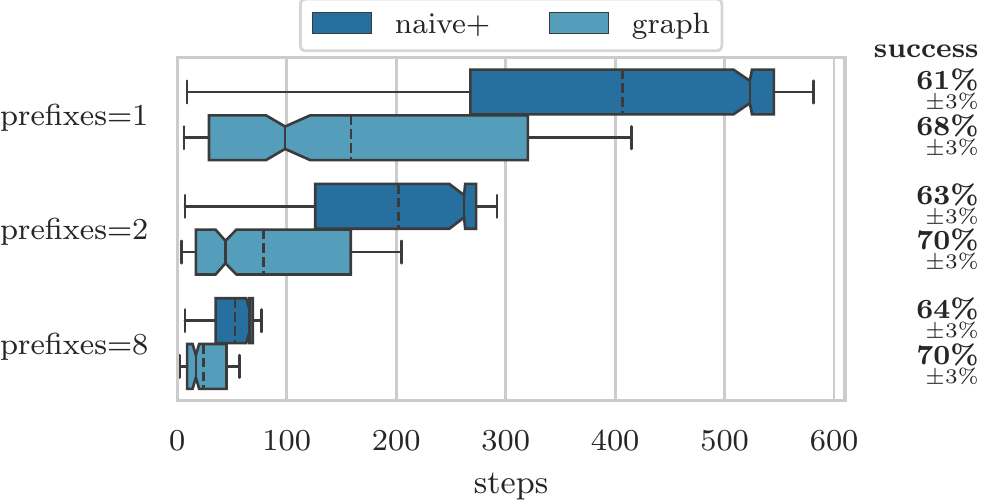}
	\caption{Influence of parallelization}
	\label{fig:eval:prefixes:steps}
\end{figure}

\subsubsection{Default Routes}
Intuitively, default routes should have a negative influence on our traceback approach, both in terms of runtime and success.
As the presence of default routes has been confirmed by multiple studies~\cite{bush2009internet,smith2020withdrawing} we would like to quantify to which extent they impact our results.
To this end, we repeated the experiment two more times, once in an ``ideal world'' setting with no default routes and once in an exaggerated setting where $80\%$ of ASes have a default route.
In line with intuition, both algorithms generally perform better with fewer default routes, as shown in \autoref{fig:eval:defaultroute:steps}.
Interestingly though, the graph-based algorithm still terminates faster when faced with $40\%$ default routes than the naive algorithm in the ideal world setting without any default routes.
The default route prevalence is also the most determining factor of traceback success:
In the idealized setting with no default routes, all simulated runs were successful, resulting in an estimated success rate of $100\%\pm 0\%$.
Yet, even with $80\%$ default routes, they achieve a success rate of $23\% \pm 3\%$ and $29\% \pm 3\%$ respectively.

\begin{figure}[t]
	\centering
	\includegraphics[width=0.99\linewidth]{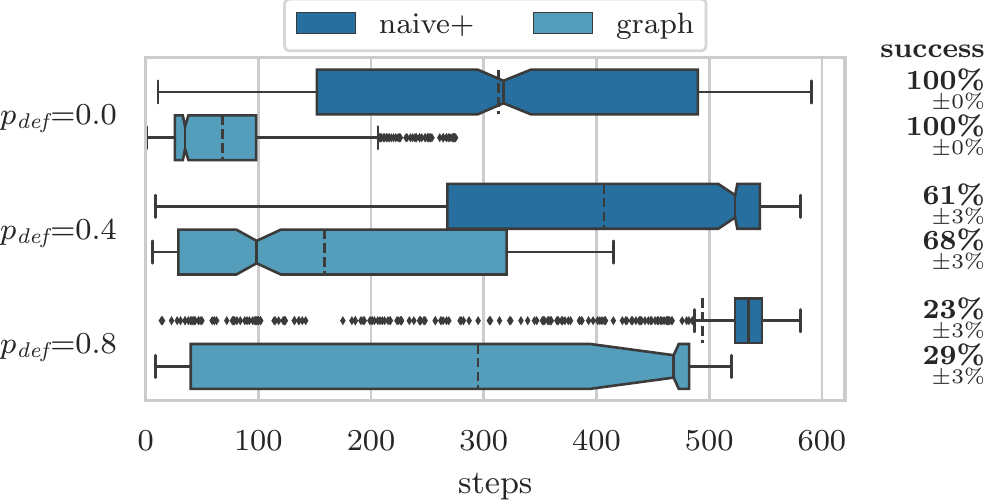}
	\caption{Influence of default route prevalence}
	\label{fig:eval:defaultroute:steps}
\end{figure}

\subsubsection{Probe Size}\label{sec:eval:sim:probesize}
We ran another experiment to measure the influence of the probe size, evaluating each algorithm with a probe size of $n=32$, $n=64$, and $n=128$.
The results, shown in \autoref{fig:eval:probesize:steps} confirm our intuition that the runtime of the naive algorithm scales inversely to the probe size.
Where at $n=128$ the naive algorithm has a median runtime of around $500$ steps, this doubles to just over $1000$ at $n=64$ and quadruples to $2000$ at $n=32$.
While the same relation holds true for the graph algorithm's \emph{maximum} runtime, $400$ steps at $n=128$ to $1600$ steps at $n=32$, it still achieves a median runtime of less than $500$ steps even with $n=32$.
At that, it outperforms both naive variants with a probe size of $n=128$.
As expected, the success rate of \toolname is not influenced by the probe size.

\begin{figure}[t]
	\centering
	\includegraphics[width=0.99\linewidth]{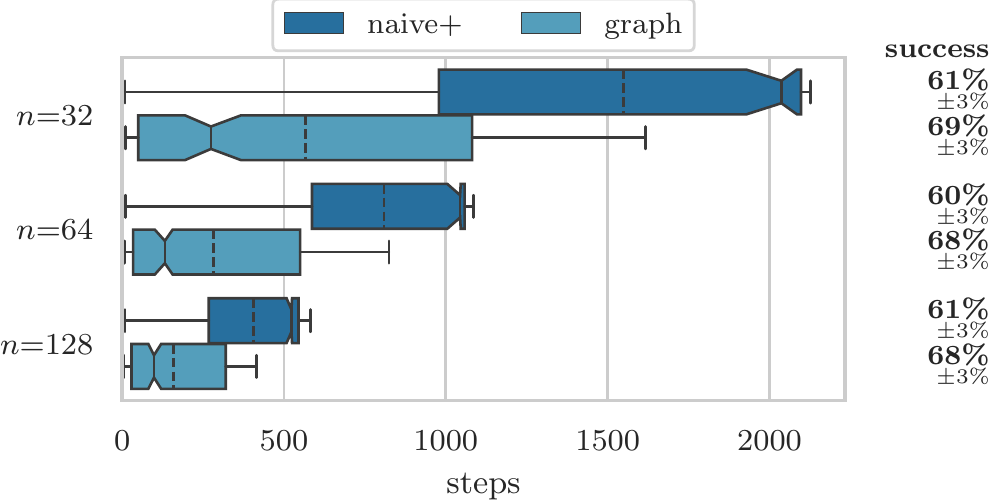}
	\caption{Influence of probe size}
	\label{fig:eval:probesize:steps}
\end{figure}

\subsubsection{Deployment Location}
To quantify the impact of the deployment location of \toolname, we also ran the experiment from three different ASes:
Next to a deployment at a PEERING customer we also simulated runs from a Tier-1 provider (AS174) as well as from a national research network.
However, as shown in \autoref{fig:eval:vantage:steps}, neither runtime nor success rate vary significantly between the three different deployments.

\begin{figure}[t]
	\centering
	\includegraphics[width=0.99\linewidth]{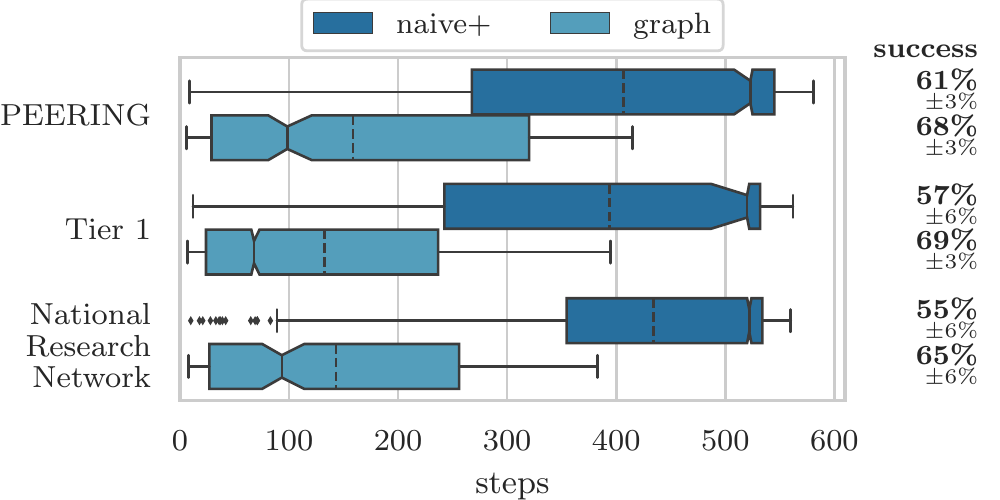}
	\caption{Influence of deployment location}
	\label{fig:eval:vantage:steps}
\end{figure}

\subsubsection{TTL Values}\label{sec:eval:sim:ttl}
\toolname relies on TTL values to detect when traffic is redirected to alternative routes.
However, while mostly stable, TTL values can change for reasons unrelated to our traceback and could even be modified by an attacker attempting to evade detection.
In a final experiment we therefore simulate how \toolname performs \emph{without} TTL values, i.e., only checking whether poisoning leads to a stop in traffic.
Perhaps surprisingly, the results in \autoref{fig:eval:ttl:steps} show that a lack of TTL values only leads to a slowdown of the graph algorithm, whose median runtime doubles, but has no statistically significant impact on overall traceback success.
This shows that \toolname can still be used even if TTL values are found to be unreliable, albeit with increased traceback times.

\begin{figure}[t]
	\centering
	\includegraphics[width=0.99\linewidth]{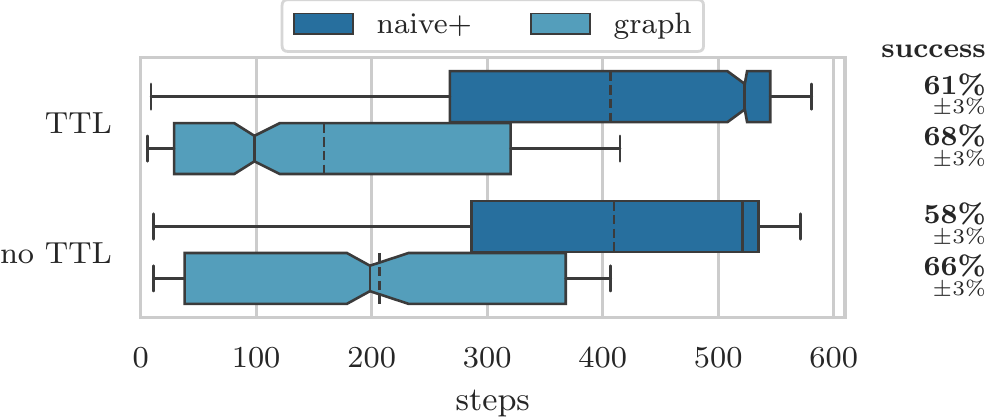}
	\caption{Influence of using TTL values}
	\label{fig:eval:ttl:steps}
\end{figure}

\subsection{Supporting Real-World Experiments}\label{sec:eval:real}
We also leveraged the PEERING BGP testbed~\cite{schlinker2019peering} to conduct experiments in the live Internet and analyzed attack data captured by the DDoS honeypot \amppot~\cite{kramer2015amppot}, to bootstrap additionally required parameters and to assess the plausibility of our simulated results.

\subsubsection{Changing Default Routes}
To faithfully model the effect of default routes, we not only need to know how many ASes \emph{have} a default route, but also \emph{how often} this default route differs from the regular route.
To measure this effect, we designed the following experiment:
From our temporary AS we would send out a poisoning advertisement for a target system, advertising one of three \texttt{/24} prefixes from our \texttt{/22} allocation.
The fourth prefix would be advertised regularly to serve as a control.
After a short while we would then ping the target system from both, the poisoning prefix and the control prefix.
If ping replies are still observed in the poisoning prefix we can conclude that the target does have a default route.
If the TTL values also differ between the poisoning and the control prefix we can further infer that it differs from the regular route.

As targets we randomly selected $624$ RIPE Atlas probes~\cite{ripeatlas} in different ASes.
TTL values were recorded five minutes after the advertisement was sent to allow routes to settle~\cite{labovitz2000delayed,katz2012lifeguard}, and new advertisements were only sent every ten minutes per prefix.
Out of the $624$ tested RIPE Atlas probes we found $360$~($58\%$) to have default routes, i.e., we would still receive pings for the prefix that was advertised with a poisoning advertisement only, a fraction larger than the default route prevalence reported by Smith et al.~\cite{smith2020withdrawing} in 2020, but lower than the one reported by Bush et al.~\cite{bush2009internet} in 2009.
We attribute the discrepancy to two effects:
First, our sample size is smaller than the one employed by both Smith et al. and Bush et al. and further biased towards ASes housing RIPE Atlas probes.
Second, we found the probe's AS information as reported by RIPE's Atlas back end to not always be accurate and thus suspect that in some cases the probe's \emph{actual} AS was different from the one poisoned, thereby making them a false positive.
As noted in \autoref{sec:eval:sim:results}, we picked a default route probability of $40\%$ for our simulator.

In $101$ of $360$~($28\%$) cases we further observed a discrepancy in TTL values between the poisoned and the control prefix, letting us conclude that in these cases the default route in fact differs from the regular route.
At a confidence level of $95\%$ this thus gives a probability of having a differing default route between $23\%$ and $33\%$.
We therefore model this in our simulator by having an AS choose a different upstream in $30\%$ of the cases when poisoned.

\subsubsection{On-Path Poisoning}
As a main primitive our approach relies on the assumption that poisoning on-path ASes provokes some observable change in the target traffic, either a change in TTLs due to a route change or a complete absence of traffic due to the lack of alternative routes.
To measure how this assumption holds up in practice we designed the following experiment:
As the ``traffic origin'' we randomly selected a RIPE Atlas probe~\cite{ripeatlas}, and used a stream of ping packets from the probe to our traceback system deployed at PEERING as the ``attack traffic''.
To obtain a real-time approximation of the taken AS path, we scheduled a traceroute measurement from the Atlas probe to our traceback system.
Running the traceroute in that direction ensures that we measure the actual ingress path to our system and do not have to assume paths to be symmetric.
Using the Team Cymru IP to ASN Lookup~\cite{cymru} we then mapped traceroute IP hops to ASes and, subsequently, poisoned every discovered on-path AS one-by-one.
In contrast to the previous experiment, this time we did not measure the impact on the poisoned AS, but on the ping packets from the Atlas probe simulating the target traffic flow.
As before we waited five minutes before taking measurements after new advertisements and ten minutes between advertisements for the same prefix.

Note that the list of on-path ASes obtained that way is naturally incomplete, as a traceroute may miss hops along the path and the IP-to-ASN mapping leads to further inaccuracies as well~\cite{zhang2010quantifying, hyun2003traceroute}.
To validate the real-time IP-to-ASN mapping we obtained from Team Cymru we later compared the mapping results to the data covering the same timeframe published on RIPEstat~\cite{ripestat}.
Here we found $22$ instances in which the IP-to-ASN mappings disagreed.
Manually analyzing these $22$ cases we determined the Team Cymru mapping to be correct in $14$ of these, and excluded the other $8$ from further analysis.
As we were only interested in measuring true on-path ASes we further also excluded hops that mapped to the same AS as the Atlas probe mimicking the traffic source.

We ran the experiment with Atlas probes located in $161$ different ASes, allowing us to collect a total of $327$ unique (on-path AS, target AS)-pairs.
In total, we found that poisoning an on-path AS resulted in a loss of traffic in $137$~($42\%$) cases and a change in the TTL in further $112$~($34\%$) cases.
Only in $78$~($24\%$) instances we found poisoning on-path ASes to have no measurable impact on the target traffic at all.

To further analyze whether the impact is related on the distance between the poisoned and the measured AS, we grouped the results by their hop distance to the target AS.
While ideally we would have used \emph{AS} hop distances for this, this would have required access to the full AS path.
We thus resorted to using \emph{IP} hop distances as they could be obtained easily from the traceroute, and let the poisoning results count towards all IP hops that were mapped to the same AS.
\autoref{fig:eval:realworld:onpath} shows the normalized results per IP hop count, with the black line denoting the number of tested AS pairs per hop distance to give an indication of their significance.
As depicted, most on-path ASes were discovered at a distance of $2$ to $10$ hops from the origin.
However, we find that overall the distance between the poisoned AS and the traffic source appears to have little impact on how traffic is affected, and that in most cases a measurable impact can be expected.

\begin{figure}
	\centering
	\includegraphics[width=0.99\linewidth]{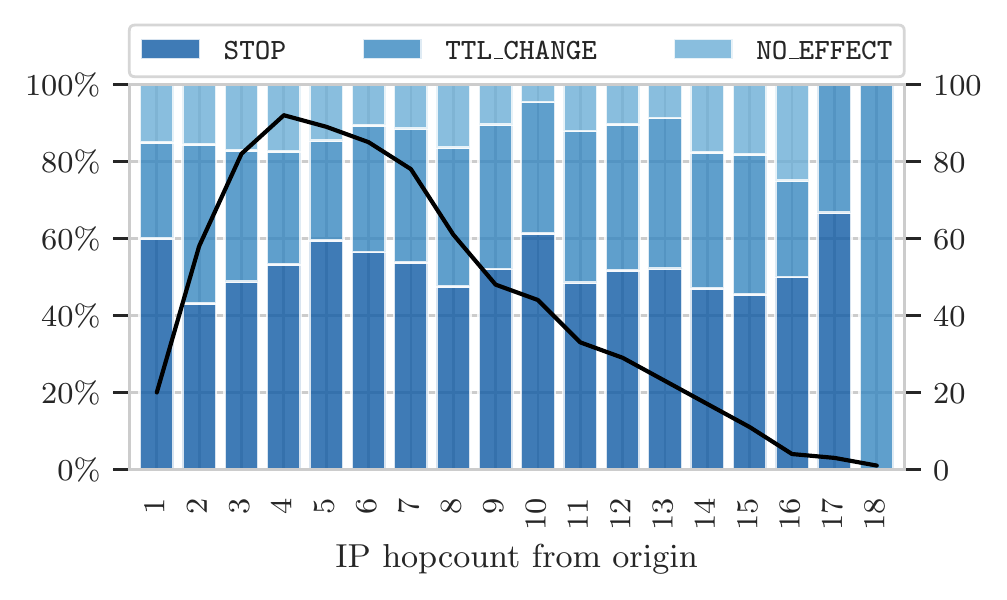}
	\caption{Real-world results of on-path AS poisoning}
	\label{fig:eval:realworld:onpath}
\end{figure}

\subsubsection{TTL Distribution}\label{sec:eval:realworld:ttl}
We leveraged the same RIPE Atlas traceroutes to obtain a realistic model of inner-AS path lengths for our simulator, such that we can simulate IP hop counts of AS paths.
For this, we utilized the same IP-to-AS mapping and then scanned the traceroute results for consecutive hops in different ASes.
If we can find two of these transitions, $A\rightarrow B$ at hops $(x, x+1)$ and $B\rightarrow C$ at hops $(y, y+1)$, then the inner-AS path length of $B$ from $A$ to $C$ is $y+1-x$ hops.
We were able to extract 137 AS-triples and their corresponding hop counts and found that a negative binomial distribution with $k=3, p=0.62$ was a good fit.

\subsubsection{Continuity of Spoofing Activity}\label{sec:eval:realworld:activity}
Even with the graph-based algorithm, \toolname has an average runtime of 147 steps, or just over one day at six advertisements per hour.
We therefore assess, how long an attack source may be observed consecutively.
To this end, we collected data on $13,321,740$ amplification attacks observed by \amppot~\cite{kramer2015amppot}.
All attacks were collected between 2015-11-25 and 2020-06-15 by a \emph{Selective Response} enabled honeypot~\cite{krupp2016identifying}.
Selective Response restricts every scanner to finding a different set of 24 of the 48 honeypot IPs, thereby imposing a unique fingerprint on the scanner.
Prior work revealed a tight connection between scan and attack infrastructure~\cite{krupp2016identifying}.
We will thus use the fingerprint as an identifier for the (unknown) traffic source.

To focus on distinctive fingerprints, we only considered attacks that used at least 12 and at most 24 honeypot IPs, which left us with $8,635,257$ remaining attacks.
For each fingerprint, we then computed the longest period, during which it was observed at least $90\%$ ($99\%$) of the time.
\autoref{fig:eval:realworld:activity} shows the fraction of \emph{attacks} whose fingerprint was observed for a given duration at a given activity level.
From this plot we find that the majority of attacks stem from sources that are also active for extended periods of time.
For example, over $68\%$ of attacks stem from a source that was seen for over a week at $90\%$ activity, $51\%$ even at $99\%$.

\begin{figure}
	\centering
	\includegraphics[width=0.99\linewidth]{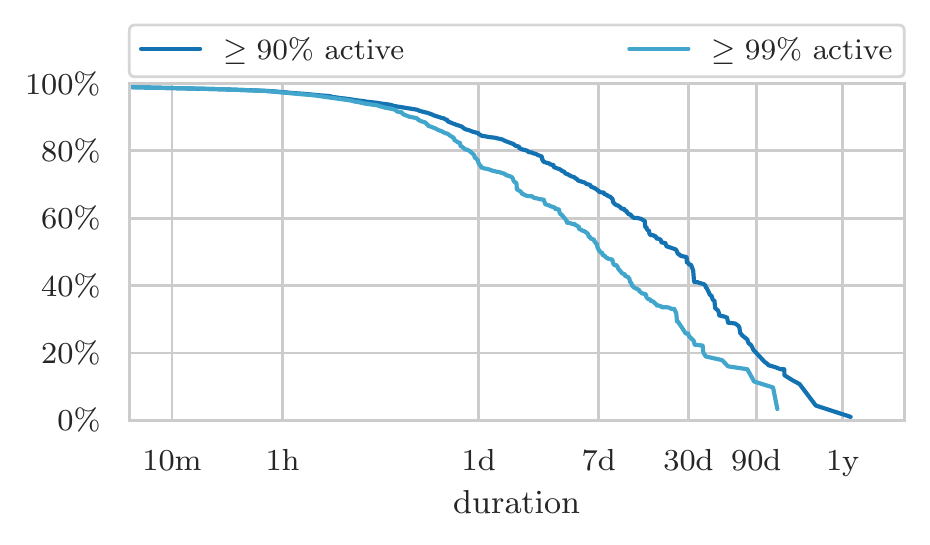}
	\caption{Fingerprint activity}
	\label{fig:eval:realworld:activity}
\end{figure}

\section{Discussion}\label{sec:discussion}
We now discuss how our evaluation results translate to the real-world deployability of our approach and the ethics of our live Internet experiments.

\subsection{BGP Mechanisms}\label{sec:discussion:bgp}
While our simulator strives to faithfully model routes and advertisement propagation, it does so by abstracting ASes and their relationships into a graph model.
This abstraction may not fully model BGP in the real world, where ASes do not act as atomic entities but advertisements are instead passed between routers and processed by actual BGP implementations.
As such, multiple mechanisms may influence the propagation of advertisements which we discuss below.

\subsubsection{Route Flap Dampening}\label{sec:discussion:rfd}
Route Flap Dampening (RFD)~\cite{rfc2439} aims to decrease the load on routers by maintaining a per-route penalty score.
This score is increased for every update, but decays exponentially over time.
Once a route's score reaches a threshold, it is no longer considered for routing nor propagated to peers until its score drops below a re-use threshold.
Studies have shown that the default thresholds used in RFD can actually be harmful even for stable routes~\cite{mao2002route}, and RFD was subsequently advised against~\cite{ripe378}.
Although later studies proposed new settings that have less adverse effects~\cite{ripe580,rfc7196,pelsser2011route}, it still remains disabled by default in, e.g., Cisco routers~\cite{ciscoBGP}.
However, as RFD maintains a score on a \emph{per-route} rather than a \emph{per-prefix} basis, it does not affect our traceback technique:
Whilst we send out multiple updates for the same prefix, every update includes a new AS path and therefore constitutes a new route~\cite[sec.~4.4.3]{rfc2439}.
Hence, even if RFD was enabled, our traceback should still work.

\subsubsection{Minimum Route Advertisement Interval}\label{sec:discussion:mrai}
Another measure to prevent high load on BGP routers is the \emph{Minimum Route Advertisement Interval} (MRAI), that limits the rate at which updates for a certain prefix are passed on to peers.
The idea behind this is that withholding routes for a certain time allows the (downstream) path exploration to converge, thereby reducing the number of updates and withdrawals sent further to peers.
The recommended value for the MRAI timer is 30 seconds~\cite[sec.~10]{rfc4271}.
As we always waited ten minutes between advertisements this should have given routers ample time for this timer to expire.
However, it also means that a real-world deployment of our approaches should employ a similar delay between advertisements.

\subsubsection{Path Filtering}\label{sec:discussion:filter}
Smith et al. also report ASes to filter advertisements based on path lengths or by checking for potential poisoning paths~\cite{smith2020withdrawing}.
For length based filters \autoref{sec:eval:sim:probesize} indicates that our graph-based approach would still perform well in many cases even when limiting the path length to $32$.
However, filtering of poisoning advertisements can impede our traceback approach if it is performed by one of the on-path ASes.
In this case, \emph{any} poisoning advertisement may provoke the target traffic to change and thus both algorithms may falsely flag the probed AS to be on-path---even if the observed change was only caused by filtering.

\subsubsection{Non-Uniform Routing}\label{sec:discussion:nonuniform}
The use of active probing in our traceback assumes that all outbound packets from an AS are affected in the same way by a route change, regardless of whether they are originating from the AS or forwarded on behalf of another.
In theory, every edge-router of an AS could behave differently and use a different route, which in turns means that different hosts in the AS could behave differently under poisoned advertisements.
A study by Mühlbauer et al.~\cite{muhlbauer2006building} finds that such route diversity can be modelled by splitting ASes into multiple \emph{quasi-routers}, each modelling a consistent routing behavior observed by the AS.
However, they find that the vast majority in route diversity comes from prefix-dependent preferences, and that ``for almost all ASes one quasi-router suffices''.
Since in our case the attack traffic as well as the active measurement replies are destined towards the \emph{same prefix}, they are not affected differently by prefix-dependent routing preferences.
We can thus conclude that, for our purposes, \emph{all} outbound traffic from an AS towards our prefix takes the same AS level path.

\subsection{Observation Correlation}\label{sec:discussion:correlation}
Our approach relies on stopping traffic and changing TTL values to infer AS-level route changes.
However, TTL values along a path may also change for other reasons (e.g., inner-AS route changes), and attack traffic may cease because the attack stopped entirely.
Therefore, if it is unclear whether a change was the result of a poisoning advertisement, the advertisement can be repeatedly withdrawn and re-advertised until a correlation can be confirmed or refuted.
Furthermore, if \toolname's control honeypots also observe the same attack traffic, they too can be used to decide whether a change is spurious.

\subsection{AS Flow Graph Correctness}\label{sec:discussion:graph}
Whereas our naive algorithm makes no assumptions about AS relationships, our graph-based traceback algorithm assumes that (1) AS paths adhere to the valley-free assumption and AS relationships follow the standard customer-provider/peer-to-peer model, and (2) a global view of these relations is available.
We discuss both assumptions in detail below.

\subsubsection{Valley-Free Assumption and AS Relationships}\label{sec:discussion:asrelation}
While both customer-provider and peer-to-peer relations between ASes are well-established and have intuitive economic incentives, the exact relation between two ASes can be arbitrarily complex.
For example, Giotsas et al.~\cite{giotsas2014inferring} identified $4026$ ASes whose relationships they classified as either \emph{hybrid} or \emph{partial transit}.
In a hybrid relation, two ASes exhibit different relations at different exchanges, whereas a partial transit relation is a restricted form of a customer-provider relation.
Giotsas and Zhou~\cite{giotsas2012valley} also find a small number of AS paths that seemingly violate the valley-free assumption, which they also attribute to non-standard AS relationships.

Our AS flow graph only captures customer-provider and peer-to-peer relations and requires paths to adhere to the valley-free assumption.
Thus, our graph-based traceback approach may falsely exclude ASes it believes to be reachable or dominated by others when faced with such non-standard relations.
However, as long as reachability and dominance information can be efficiently encoded (e.g., through adapting links in the graph), a similar algorithm may still be employed.

\subsubsection{AS Relationship Dataset}
AS relationships are usually subject to non-disclosure agreement and can thus only be inferred from publicly available routing data.
While the state-of-the-art of inferring the global AS relationship graph has been constantly evolving~\cite{gao2001inferring,subramanian2002characterizing,di2003computing,erlebach2002classifying,xia2004evaluation,dimitropoulos2005inferring,dimitropoulos2007relationships,hummel2007acyclic,luckie2013relationships}, inference will inevitably only be able to produce an approximation of the AS graph.
As with non-standard AS relations, missing or incorrect links are problematic for our graph-based algorithm, which could cause it to wrongly discard ASes as candidates.
In that regard, our simulation results should be seen as best-case results, as both simulator and traceback use the same graph data.

\subsection{Induced Route Changes}\label{sec:discussion:induced}
In some cases, a poisoned advertisement can lead to route changes or losses even at ASes that can never receive an advertisement through one of the poisoned ASes.
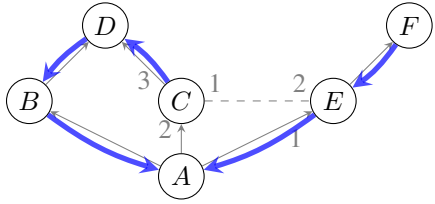
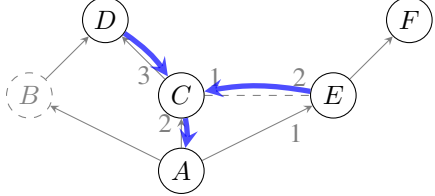
\begin{figure}[t]
 \centering
 \begin{subfigure}[t]{\columnwidth}
 \centering
 \begin{tikzpicture}[]

  \node [circle,draw,minimum size=6mm,label=center:$A$] (a) at (0,-2) {};
  \node [circle,draw,minimum size=6mm,label=center:$B$] (b) at (-2,-1) {};
  \node [circle,draw,minimum size=6mm,label=center:$C$] (c) at (0,-1) {};
  \node [circle,draw,minimum size=6mm,label=center:$D$] (d) at (-1,0) {};
  \node [circle,draw,minimum size=6mm,label=center:$E$] (e) at (2,-1) {};
  \node [circle,draw,minimum size=6mm,label=center:$F$] (f) at (3,0) {};

  \draw[->,>=stealth, opacity=0.5] (a) -- (b);
  \draw[->,>=stealth, opacity=0.5] (a) -- node[left, pos=0.85] {2} (c);
  \draw[->,>=stealth, opacity=0.5] (a) -- node[below, pos=0.85] {1} (e);
  \draw[dashed, opacity=0.5] (c) -- node[above, pos=0.10] {1} node[above, pos=0.90] {2} (e);
  \draw[->,>=stealth, opacity=0.5] (b) -- (d);
  \draw[->,>=stealth, opacity=0.5] (c) -- node[left, pos=0.10] {3} (d);
  \draw[->,>=stealth, opacity=0.5] (e) -- (f);
	\draw [->, >=stealth, blue, opacity=0.7, line width=2pt] (b) to[bend right=10] (a);
	\draw [->, >=stealth, blue, opacity=0.7, line width=2pt] (d) to[bend right=10] (b);
	\draw [->, >=stealth, blue, opacity=0.7, line width=2pt] (c) to[bend right=10] (d);
	\draw [->, >=stealth, blue, opacity=0.7, line width=2pt] (e) to[bend left=10] (a);
	\draw [->, >=stealth, blue, opacity=0.7, line width=2pt] (f) to[bend left=10] (e);
 \end{tikzpicture}
	 \caption{$A$ is customer of $B$, $C$, and $E$; $B$ and $C$ are customers of $D$; $E$ is customer of $F$; $D$ and $E$ have a peer-to-peer relation; numbers indicate local preferences; thick arrows show the resulting routes to $A$.}
 \end{subfigure}

 \begin{subfigure}[t]{\columnwidth}
 \centering
 \begin{tikzpicture}[]
	\node [circle,draw,minimum size=6mm,label=center:$A$] (a) at (0,-2) {};
	 \node [circle,draw,minimum size=6mm,label={[opacity=0.5]center:$B$},opacity=0.5,dashed] (b) at (-2,-1) {};
  \node [circle,draw,minimum size=6mm,label=center:$C$] (c) at (0,-1) {};
  \node [circle,draw,minimum size=6mm,label=center:$D$] (d) at (-1,0) {};
  \node [circle,draw,minimum size=6mm,label=center:$E$] (e) at (2,-1) {};
  \node [circle,draw,minimum size=6mm,label=center:$F$] (f) at (3,0) {};

  \draw[->,>=stealth, opacity=0.5] (a) -- (b);
  \draw[->,>=stealth, opacity=0.5] (a) -- node[left, pos=0.85] {2} (c);
  \draw[->,>=stealth, opacity=0.5] (a) -- node[below, pos=0.85] {1} (e);
  \draw[dashed, opacity=0.5] (c) -- node[above, pos=0.10] {1} node[above, pos=0.90] {2} (e);
  \draw[->,>=stealth, opacity=0.5] (b) -- (d);
  \draw[->,>=stealth, opacity=0.5] (c) -- node[left, pos=0.10] {3} (d);
  \draw[->,>=stealth, opacity=0.5] (e) -- (f);
	\draw [->, >=stealth, blue, opacity=0.7, line width=2pt] (c) to[bend left=10] (a);
	\draw [->, >=stealth, blue, opacity=0.7, line width=2pt] (d) to[bend left=10] (c);
	\draw [->, >=stealth, blue, opacity=0.7, line width=2pt] (e) to[bend right=10] (c);

 \end{tikzpicture}
 \caption{Poisoning $B$ causes $C$ to switch routes, which induces a route change at $E$ and a loss of connection at $F$.}
 \end{subfigure}
 \caption{Example of induced changes}\label{fig:induced}
\end{figure}
Consider for example the network shown in \autoref{fig:induced}.
In the normal state, AS $C$ receives three advertisements for routes to $A$, the direct route from $A$, the route ${D}\to{B}\to{A}$ from $D$, and the route ${E}\to{A}$ from $E$.
From these it will choose the route via $D$, since it has the highest local preference.
However, since this route is received from one of $C$'s providers, it cannot be exported to the peer $E$.
AS $E$ therefore only sees the direct route from $A$, which it can further advertise to its provider $F$.

Poisoning $B$ makes the route ${D}\to{B}\to{A}$ unavailable.
$C$ therefore switches to its next preferred route, the direct route to $A$.
Since $A$ is a customer of $C$, $C$ can advertise this new route ${C}\to{A}$ to its peer $E$.
This, however, induces a route change at $E$, because this new route has a higher local preference at $E$.
Furthermore, since the best route at $E$ is now received from a peer, it can no longer be exported to $E$'s provider $F$.
$F$ therefore loses its connection to $A$.
Note that neither $E$ nor $F$ could ever have a route to $A$ via the poisoned AS $B$.
Yet, they see a route change or even a loss of connectivity.

An AS can only cause such \emph{induced} changes, if it can receive two different routes, one from a customer and one from a peer or provider.
Only in that case, the set of other ASes that it can export routes to can change:
Switching from a customer-provided route to a peer/provider-provided one limits it to advertise this route to its customers, switching the other way enables it to also advertise a route to its peers and providers.
We will call these ASes \emph{ambiguous}.
ASes that are reachable through an ambiguous AS may therefore experience induced changes.
Further, if an AS is \emph{only} reachable through peers or providers of ambiguous ASes, it may also experience an induced connectivity loss.

For our naive algorithm, such induced changes can cause additional ASes to appear in the final result set (e.g., in the example above, $B$ would be erroneously considered \emph{on-path} even if the attack came from $F$).
Yet, these induced changes cannot ``hide'' actual on-path ASes from detection.
Our graph algorithm on the other hands requires a small modification (shown in \autoref{alg:traceback:induced}) to correctly handle induced changes:
Whenever the candidate set is reduced, ASes that could show the observed behaviour due to induced changes need to be retained.
\begin{figure}[!t]
 \begin{algorithmic}
    \Procedure{UpdatePassive'}{$\probe, \rp$}
      \If{$\rp = \trafficstop$}
			  \State $\candidates_{stop} \gets \indStop_{G_A}(\probe)$
				\State $\candidates \gets \candidates \cap \left(\dominatees_{G_A}(\probe) \cup \candidates_{stop}\right)$
      \ElsIf{$\rp = \ttlchange$}
			  \State $\candidates_{change} \gets \indChange_{G_A}(\probe)$
				\State $\candidates \gets \candidates \cap \left(\reachable_{G_A}(\probe) \cup \candidates_{change}\right)$
      \EndIf
    \EndProcedure
    \Statex
    \Procedure{UpdateActive'}{$\probe, \rp, \ra$}
      \If{$\rp = \trafficstop$}
        \State $\exclude \gets \left\{X \in \probe\mid\ra[X] \neq \trafficstop\right\}$
			  \State $\candidates_{stop} \gets \indStop_{G_A}(\probe)$
				\State $\candidates \gets \candidates \setminus \left(\reachable{G_A}(\exclude) \setminus \candidates_{stop}\right)$
      \Else
        \State $\exclude \gets \left\{X \in \probe\mid\ra[X] = \trafficstop\right\}$
			  \State $\candidates_{change} \gets \indChange_{G_A}(\probe)$
				\State $\candidates \gets \candidates \setminus \left(\dominatees_{G_A}(\exclude) \setminus \candidates_{change}\right)$
      \EndIf
      \State $\probe \gets \probe \setminus \exclude$
    \EndProcedure
 \end{algorithmic}
 \caption{Flow Graph based traceback algorithm adapted to handle induced route changes}\label{alg:traceback:induced}
\end{figure}

\begin{figure}[t]
	\centering
	\includegraphics[width=0.99\linewidth]{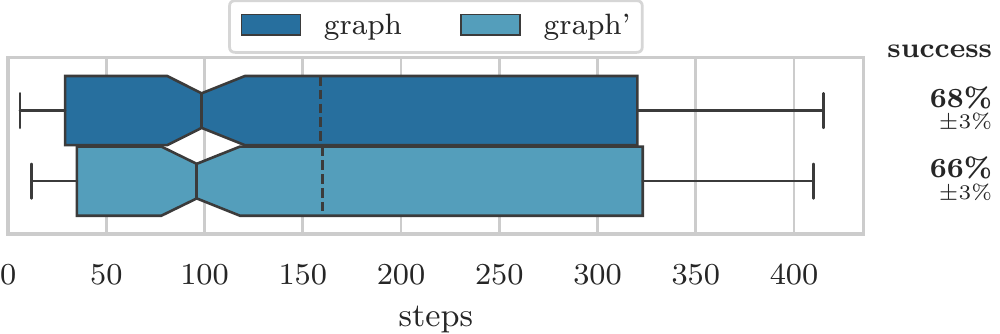}
	\caption{Influence of induced route changes. The modified algorithm is labeled \emph{graph'}.}
	\label{fig:eval:induced:steps}
\end{figure}
To assess the impact of induced route changes, we ran a simulation of our graph-based traceback with these modifications.
As shown in \autoref{fig:eval:induced:steps}, neither runtime nor success rate differ significantly compared to the unmodified version.
This can be explained because ambiguous ASes are relatively rare:
Analyzing the flowgraph used during simulation reveals only $231$ ambiguous ASes.
We can thus conclude, that induced route changes only have negligible impact on the traceback performance.

\subsection{Multi-Source Attacks}
Amplification attacks are largely launched from single sources~\cite{kramer2015amppot,krupp2017linking}.
For \toolname we thus assume that every attack has a unique source AS.
In theory though, amplification attacks could also be launched from multiple colluding sources in different ASes.

If the attack traffic of a multi-source attack can be separated by source (e.g., by different TTL values due to different path lengths), both algorithms of \toolname still work as before.
Otherwise, successful poisoning of one of the sources can only be measured as a decrease in attack traffic volume.
Yet, even in that case, the naive algorithm should be able to find all attack sources.

\subsection{Evasion}\label{sec:discussion:evasion}
During a traceback run \toolname creates a large number of route advertisements for the probing prefixes.
Attackers that are aware of our system may therefore try to evade it by monitoring public BGP data~\cite{routeviews,riperis} in order to identify and exclude the probing prefixes.
While we cannot \emph{prevent} such active evasion attempts, we can at least \emph{detect} if an attacker is evading our system.
In this case, we would not observe any attacks from a given adversary at the probe honeypots, but would keep observing them at the control honeypots.
Yet, even outside of \toolname, amplification honeypots are inherently detectable due to their rate limiting behavior~\cite{kramer2015amppot}.

Besides evasion, attackers might also attempt to deceive \toolname by modifying their initial TTL values.
However, as shown in \autoref{sec:eval:sim:ttl}, \toolname still performs well even when ignoring TTL values.

\subsection{Ethical Considerations}
We took several measures to ensure that our experiments did not impact other systems or lead to instabilities in the BGP.
To obtain a real-time estimate of the currently active BGP path, we conducted traceroute runs from RIPE Atlas probes to our measurement system hosted at PEERING.
In order to keep the impact on other systems minimal, we only used one-off measurements with the default values defines by RIPE (i.e., packets of 48 bytes, at most 3 packets per destination).
For our active ping measurements we ensured to send at most one packet per minute per target on average.
At 64 bytes per packet (1Bps), we believe these to have negligible impact.

All experiments that involved sending (poisoned) BGP advertisements were conducted only after consulting PEERING operators.
To further ensure that other experiments running at the testbed were not influenced by ours, we used temporary prefixes and a temporary ASN allocated by our regional Internet registry for the purpose of these experiments.
We also registered our temporary allocations in the WHOIS database, such that network operators were able to contact us directly.
Additionally, we closely monitored network operator mailing lists.

\subsubsection{Impact on Legitimate Traffic}
Poisoning advertisements can render a prefix temporarily unreachable from parts of the Internet.
Therefore, \toolname's probing prefixes should host no other systems but honeypot reflectors.
As we use BGP Poisoning on these small prefixes only, other prefixes remain unaffected---effectively excluding collateral damage on benign traffic.

\section{Related Work}\label{sec:related}
We find related work from two fields of study:
The first considers the problem of IP spoofing and traceback, the second is concerned with BGP Poisoning as a primitive for traffic engineering and security applications.

\subsection{IP Spoofing and Traceback}
IP spoofing and the resulting need for traceback has been an active field of research for many years.
One common approach collects (statistical) telemetry data at multiple routers from which the origin of spoofed packets can later be derived~\cite{snoeren2001hash,snoeren2002single,li2004large,sung2008large,korkmaz2007single}.
Another approach lets routers encode path information into the packet itself~\cite{savage2000practical,doeppner2000using}, using additional IP options or unused header fields~\cite{savage2000practical,song2001advanced,dean2002algebraic,yaar2006stackpi,chen2007divide}, advanced encoding schemes~\cite{song2001advanced,yaar2003pi,gao2007practical}, or probabilistic techniques~\cite{savage2000practical,savage2001network,duwairi2004efficient,dong2005efficient,shokri2006ddpm} to reduce the per-router overhead.
In theory, both approaches could perfectly track the origin of spoofed traffic.
However, both require a widespread deployment in routers and the cooperation of multiple ISPs.
As some of these techniques have been proposed almost two decades ago, it is clear that this is a inhibiting factor for traceback in practice.
In contrast, \toolname requires no cooperation of other systems nor changes to existing routers.

Another line of work considers the problem of traceback in the specific case of amplification attacks.
Krupp et al.~\cite{krupp2016identifying,krupp2017linking} relaxed the traceback problem to finding scanning systems used for attack preparation or re-identifying Booter services responsible for attacks.
In contrast, our approach can identify ASes actively participating in the attack without upfront knowledge of the system.

A third line of research considers the broader question of identifying systems that are \emph{capable} of IP spoofing~\cite{beverly2005spoofer,beverly2013initial,kuhrer2014exit,luckie2019network}.
While another important step in alleviating the problem of IP spoofing, we consider their work orthogonal to ours:
Whereas they find systems that could in principle send spoofed packets, we aim to identify malevolent actors red-handed.

In concurrent work, Fonseca et al.~\cite{fonseca2020tracking} also attempt to identify spoofing sources by varying BGP advertisements from multiple anycast locations.
For this, they create $700$ different advertisement configurations by selecting subsets of their peers, adding path prepends, and poisoning immediate neighbors.
For every configuration they then record through which peer the attack traffic is received, thereby generating a fingerprint of source ASes.
In contrast to \toolname, their approach necessarily requires $700$ probing steps and cannot distinguish ASes that share a common path beyond their immediate neighbors.

\subsection{BGP Poisoning}
Although only a side-effect of BGP's loop detection mechanism, BGP Poisoning has seen a wide field of applications.
In the area of measurement studies, Colitti et al.~\cite{colitti2007investigating} and Anwar et al.~\cite{anwar2015investigating} employed BGP Poisoning to supplement the analysis of AS relationships and prefix propagation.
Katz-Bassett et al.~\cite{katz2011machiavellian,katz2012lifeguard} showed how BGP Poisoning may be used to actively repair routes, and Smith et al.~\cite{smith2018routing} later showed how it can be used to avoid DDoS congested links.
Finally, works by Tran et al.~\cite{tran2019feasibility} and Smith et al.~\cite{smith2020withdrawing} analyze how well BGP Poisoning works in practice through extensive measurements.
Yet, to the best of our knowledge, our approach and the concurrent work by Fonseca et al.~\cite{fonseca2020tracking} are the first to leverage BGP Poisoning for DDoS attack traceback.

\section{Conclusion}\label{sec:conclusion}
IP spoofing not only enables amplification attacks, but also hides the attackers' true whereabouts.
Our system \toolname employs a novel traceback approach that shows that BGP Poisoning can be used to track down an attacker's network location---requiring neither the assistance of external parties nor knowing the attacker in advance.
We find that our naive algorithm has a median runtime of $549$ steps, or just under four days with realistic parameters, thus showing the feasibility of our approach in practice.
Our second algorithm leverages a graph model of BGP path propagation built from AS relationship data and manages to reduce this runtime to $98.5$ steps, or just over one day, for the same parameters---and in only $29$ steps, under five hours, in a quarter of all cases.

\ifthenelse{\boolean{blind}}{
}{
	\section*{Acknowledgment}
We would like to thank PEERING for letting us conduct real-world BGP measurements as well as RIPE NCC for a temporary IP-prefix and ASN allocation and the RIPE Atlas platform.
We would further like to thank Ethan Katz-Bassett and Italo Cunha for an insightful discussion of an earlier draft of this paper.
Finally, we would also like to thank the anonymous reviewers and our shepherd for their valuable feedback.

}

\bibliographystyle{ieeetr}
\bibliography{biblio,rfcs}

\clearpage
\appendices
\section{Correctness of AS Flow Graphs}\label{app:graph:proof}
\begin{theorem}
	Let $\aspath = (X_1, \dots, X_n) \in \AS^n$ be a valley-free path from AS $X_1$ to AS $X_n$ and $G = (V,E)$ the flow graph constructed according to \autoref{sec:graph:flowgraph}.
	Then there exists a path $\nodepath = (x_1 = u_{X_1}, \dots, x_m = c_{X_n}) \in V^m$ in $G$.
\end{theorem}

\begin{proof}
	Since $\aspath = (X_1, \dots, X_n)$ is a valley-free path, it can be split into three (possibly empty) parts: an ``uphill'' prefix, a single ``peak'' peer-to-peer link, and a ``downhill'' suffix.
	Formally, $\exists k,l, 1\leq k \leq l \leq k+1 \leq n$ s.t.:
	\begin{enumerate}
		\item $\forall 1 \leq i < k: (X_i, X_{i+1}) \in \CP$
		\item $k < l \implies (X_k, X_l) \in \PP$ 
		\item $\forall l \leq i < n: (X_{i+1}, X_i) \in \CP$
	\end{enumerate}
	For the first part, we can find a path in $G$ as by construction it holds that $\forall 1 \leq i < k: (u_{X_i}, u_{X_{i+1}}) \in E$.
	Likewise, for the last part, since $(X_{i+1}, X_i) \in \CP$ it holds that $\forall l \leq i < n: (c_{X_i}, c_{X_{i+1}}) \in E$.
	If the path has a peer-to-peer link, i.e., $k < l$, then $(u_{X_k}, c_{X_l}) \in E$.
	Otherwise, it holds that $X_k = X_l$ and thus $(u_{X_k}, c_{X_l}) = (u_{X_k}, c_{X_k})\in E$.
	Therefore, the path $\nodepath = (u_{X_1}, \dots, u_{X_k}, c_{X_l}, \dots, c_{X_n})$ is in G and satisfies the requirements.
\end{proof}

\begin{theorem}
	Let $G = (V,E)$ be the flow graph constructed according to \autoref{sec:graph:flowgraph} and $\nodepath = (x_1, \dots, x_m) \in V^m$ be a path in $G$.
	Then there exists a valley-free path $\aspath = (X_1 = \asn(x_1), \dots, X_n = \asn(x_m)) \in \AS^n$.
\end{theorem}

\begin{proof}
	W.l.o.g. assume that $x_1 = u_{X_1}$ and $x_m = c_{X_n}$.
	Since all edges in $E$ are of the form $(u_A, u_B)$, $(c_A, c_B)$, or $(u_A, c_B)$ the path can only have one transition from ``unconstrained'' to ``constrained'' nodes, i.e., $\exists 1 \leq i < n$ such that
	\begin{enumerate}
		\item $\forall 1 \leq j \leq i : x_j \in \left\{u_A | A \in \AS \right\}$
		\item $\forall i < j \leq n : x_j \in \left\{c_A | A \in \AS \right\}$
	\end{enumerate}
	Therefore, for $j < i$ it holds that $(\asn(x_j), \asn(x_{j+1})) \in \CP$, as otherwise there would be no edge in $G$ between them.
	Likewise, for $j > i$ it holds that $(\asn(x_{j+1}), \asn(x_j)) \in \CP$.
	Finally, for the edge $(x_i, x_{i+1}) = (u_{X_k}, c_{X_l})$ it must hold that either $X_k = X_l$ or that $(X_k, X_l) \in \PP$.
	Therefore, $\aspath = (\asn(x_1), \dots, \asn(x_m))$ follows the definition of a valley-free path, with the exception that it may still contain loops.
	However, if $\aspath$ contains a loop, i.e., $\exists 1 \leq i < j \leq n$ s.t. $\aspath = (X_1, \dots, X_i, \dots, X_j, \dots X_n)$ with $X_i = X_j$, one can easily construct a loop-free path $\aspath' = (X_1, \dots, X_i, \dots X_n)$ by omitting positions $i+1$ through $j$.
	Note that removal of loops does not affect the path's endpoints, and thus the resulting AS path $\aspath'$ is a valley-free path from $\asn(x_1)$ to $\asn(x_m)$.
\end{proof}

\end{document}